\documentclass[11pt]{article}

\usepackage[margin=1in]{geometry}
\usepackage{amsmath,array,bm}
\usepackage{amsmath,amsfonts}
\usepackage{amsthm}
\usepackage{amssymb,latexsym}
\usepackage{algorithm}
\usepackage{subcaption}
\usepackage{algorithmicx}
\usepackage{marginnote}
\usepackage[]{algpseudocode}
\usepackage[dvipsnames]{xcolor}
\usepackage{graphicx}
\usepackage{wrapfig, framed, caption}
\usepackage{caption}
\usepackage{float}
\usepackage{thmtools}
\usepackage{hyperref}
\usepackage{mathtools}
\usepackage{thm-restate}
\usepackage{cleveref}
\usepackage{bussproofs}
\usepackage[square,sort,comma,numbers]{natbib}
\usepackage{tikz}
\usepackage{xspace}
\usepackage{xcolor}
\usepackage{tipa}
\usepackage{mdframed}
\usepackage[export]{adjustbox}[2011/08/13]
\usepackage{multirow}
\usepackage{hhline}
\usepackage{todonotes}
\usepackage{tcolorbox} 
\usepackage{dirtytalk}

\newtheorem{lemma}{Lemma}
\newtheorem{claim}{Claim}
\newtheorem{observation}{Observation}
\newtheorem{theorem}{Theorem}
\newtheorem{corollary}{Corollary}
\newtheorem{definition}{Definition}
\newtheorem{proposition}{Proposition}

\newcommand{\difty}[1]{D(#1)}
\newcommand{\nifty}[1]{||#1||_{\infty}}
\newcommand{\norm}[1]{||#1||}
\newcommand{\dx}{\Delta_x}
\newcommand{\dy}{\Delta_y}
\newcommand{\seg}[1]{\overline{#1}}

\title{Coresets for Clustering in Geometric Intersection Graphs}
\author{Sayan Bandyapadhyay \thanks{Portland State University, Oregon, USA.} \and Fedor V. Fomin \thanks{University of Bergen, Norway.} \and Tanmay Inamdar \addtocounter{footnote}{-1}\footnotemark{} \thanks{The research leading to these results has received funding from the Research Council of Norway via the project BWCA (grant no. 314528), and the European Research Council (ERC) via grant LOPPRE, reference 819416.}}
\date{} 

\newcommand{\real}{\mathbb{R}}

\renewcommand{\d}{\mathsf{d}}
\renewcommand{\P}{\mathcal{P}}
\newcommand{\cost}{\mathsf{cost}}

\newcommand{\lr}[1]{\left( #1\right)}
\newcommand{\LR}[1]{\left\{ #1\right\}}
\newcommand{\Oh}{\mathcal{O}}
\newcommand{\A}{\mathcal{A}}
\newcommand{\T}{\mathcal{T}}
\newcommand{\W}{\mathcal{W}}
\newcommand{\tils}{\tilde{s}}
\newcommand{\R}{\mathcal{R}}
\renewcommand{\ij}{^{i, j}}

\renewcommand{\L}{\mathcal{L}}
\DeclarePairedDelimiter\near{\lceil}{\rfloor}
\DeclareMathOperator*{\argmin}{arg\,min}

\newcommand{\calS}{\mathcal{S}}
\newcommand{\tilS}{\tilde{\calS}}
\newcommand{\bbC}{\mathbb{C}}
\newcommand{\bbCL}{\mathbb{C}_{\text{landmark}}}
\newcommand{\bbCS}{\mathbb{C}_{\text{support}}}
\newcommand{\bbCN}{\mathbb{C}_{\text{net}}}

\newcommand{\Snet}{\mathcal{S}_{\text{net}}}
\newcommand{\Snetsub}{\mathcal{S}_{\text{net-sub}}}
\newcommand{\Ssup}{\mathcal{S}_{\text{support}}}
\newcommand{\Sland}{\mathcal{S}_{\text{landmark}}}

\begin{document}

\maketitle

\begin{abstract}
Designing coresets---small-space sketches of the data preserving cost of the solutions within $(1\pm \epsilon)$-approximate factor---is an important research direction in the study of center-based $k$-clustering problems, such as $k$-means or  $k$-median.  
Feldman and Langberg [STOC'11] have shown that for $k$-clustering of $n$ points in general metrics, it is possible to obtain coresets whose size depends logarithmically in $n$. Moreover,  such a dependency in $n$ is inevitable in general metrics. A significant amount of  recent work in the area is devoted to obtaining coresests whose sizes are independent of $n$ for special metrics, like  $d$-dimensional Euclidean space [Huang, Vishnoi, STOC'20],  doubling metrics [Huang, Jiang, Li, Wu, FOCS'18], metrics of graphs of bounded treewidth [Baker, Braverman, Huang, Jiang, Krauthgamer, Wu, ICML’20], or graphs excluding a fixed minor [Braverman, Jiang, Krauthgamer, Wu, SODA’21].

In this paper, we provide the first constructions of coresets whose size does not depend on $n$ for $k$-clustering in the metrics induced by \emph{geometric intersection graphs}. For example, we obtain $\frac{k\log^2k}{\epsilon^{\Oh(1)}}$ size coresets for $k$-clustering in Euclidean-weighted unit-disk graphs (UDGs) and unit-square graphs (USGs). 
These constructions follow from a general theorem that identifies two canonical properties of a graph metric sufficient for obtaining coresets whose size is independent of $n$. The proof of our  theorem builds on the recent work of 
Cohen-Addad, Saulpic, and Schwiegelshohn [STOC '21], which ensures small-sized coresets conditioned on the existence of an interesting set of centers, called \emph{centroid set}. The main technical contribution of our work is the proof of the existence of such a small-sized centroid set for graphs that satisfy two canonical geometric properties. Loosely speaking, we exploit the fact that the metrics of geometric intersection graphs are ``similar'' to the Euclidean metrics for points that are close, and to the shortest path metrics of planar graphs for points that are far apart. The main technical challenge in constructing centroid sets of small sizes is in combining these two very different metrics. 

The new coreset construction helps to design the first $(1+\epsilon)$-approximation for center-based clustering problems in UDGs and USGs, that is fixed-parameter tractable in $k$ and $\epsilon$ (FPT-AS). 
\end{abstract}

\section{Introduction}
\label{sec:intro}
Clustering is one of the most important data analysis techniques where the goal is to partition a dataset into a number of groups such that each group contains similar set of data points. The notion of similarity is captured by a distance function between the data points, and the goal of retrieving the {best} natural clustering of the data points is achieved by minimizing a proxy cost function. 
In this work, we study the popular $(k,z)$-clustering problem.

\medskip\noindent\textbf{$(k,z)$-clustering.} Given a set of points $P$ in a metric space $(\Omega,\d)$ and two positive integers $k$ and $z$, find a set $C$ of $k$ points (or centers) in $\Omega$ that minimizes the following cost function: \[\cost(C)=\sum_{p\in P} \cost(p, C)\] where $\cost(p, C)=(\d(p,C))^z$ and $\d(p,C)=\min_{c\in C} \d(p,c)$. 

Two widely studied clustering problems, $k$-means, and $k$-median clustering, are  special versions of $(k,z)$-clustering with $z=2$ and $z=1$, respectively. 
  A popular way of dealing with large data for the purpose of the analysis is to apply a data reduction scheme as a preprocessing step. In the context of clustering, one such way of preprocessing the data is to construct an object known as \emph{coresets}.

\medskip\noindent\textbf{Coresets.} Informally, an $\epsilon$-coreset for $(k,z)$-clustering is a small-sized summary of the data that approximately (within $(1\pm \epsilon)$ factor) preserves the cost of clustering with respect to any set of $k$ centers (we will often shorten ``$\epsilon$-coreset'' to simply ``coreset''). Thus, any solution set of centers on the coreset points can be readily used as a solution for the original dataset. A formal definition follows.
\begin{definition}[$\epsilon$-Coreset]
	A coreset for $(k, z)$-clustering problem on a set $P$ of points in a metric space $(\Omega, \d)$ is a weighted subset $Y$ of $\Omega$ with weights $\omega: Y \to \real^+$ such that for any set $\calS \subseteq \Omega$ with $|\mathcal{S}| = k$, 
	\[ \left| \sum_{p \in P} \cost(p, \mathcal{S}) - \sum_{p \in Y} \omega(p) \cost(p, \mathcal{S}) \right| \le \epsilon \cdot \sum_{p \in P}  \cost(p, \mathcal{S}). \]
\end{definition}

Feldman and Langberg \cite{DBLP:conf/stoc/FeldmanL11} showed that for $n$ points in any general metric, a coreset of size $\Oh(\epsilon^{-2z}k\log k\log n)$ can be constructed in time $\tilde{\Oh}(nk)$, where $\tilde{\Oh}()$ notation hides a poly-logarithmic factor.
 Also, it is known that the dependency on $\log n$ in the above bound cannot be avoided \cite{DBLP:conf/icml/BakerBHJK020,Cohen-AddadLSS22}. 
However, for several special metrics, it is possible to construct coresets  whose size does not 
depend on the data size.  There has been a large pool of work for Euclidean spaces, culminating in a bound of $\Oh(k\cdot (\log k)^{\Oh(1)}\cdot 2^{\Oh(z\log z)}\epsilon^{-2} \cdot \min\{\epsilon^{-z},k\})$ \cite{Cohen-AddadLSS22}, which is independent of the data size and dimension of the space. Moreover, the question has been studied in other specialized settings such as doubling metrics \cite{huang2018epsilon}, shortest-path metrics in the graphs of bounded-treewidth \cite{DBLP:conf/icml/BakerBHJK020}, and graphs excluding a fixed minor \cite{braverman2021coresets}. A recent result by Cohen-Addad et al.\ \cite{Cohen-AddadLSS22} gives a unified framework that encompasses all these results. We note that, since the number of distinct (weighted) points in coresets is usually much smaller (and sometimes independent of) $n$, they naturally find applications in non-sequential settings such as streaming \cite{DBLP:conf/stoc/Har-PeledM04,chen2009coresets}.

Let us remark that all known results about small coresets in graph metrics strongly exploit the sparsity property of graphs such as bounded treewidth~\cite{DBLP:conf/icml/BakerBHJK020} or excluding a fixed minor~\cite{braverman2021coresets}. There is a very good reason for that.  In a complete graph,  by setting suitable weights on the edges one can represent \emph{any} general metric.  Thus if a graph family contains large cliques, clustering in such graphs is as difficult as in general metrics.

In this work, we are interested in coreset construction for edge-weighted geometric intersection graphs with shortest-path metric. A geometric intersection graph of a set of geometric objects contains a vertex for each object and an edge corresponding to each pair of objects that have non-empty intersection. (We note that for our purpose of designing algorithms, we do not explicitly need the objects or their geometric representation. It is sufficient to work with the graph representation as long as the edge-weights are given.) In particular, geometric intersection graphs are a widely studied model for ad-hoc communication and wireless sensor networks \cite{schulz2007modeling,balakrishnan2004distance,lebhar2009unit,li2005efficient,kuhn2005local}. Notably, clustering is a common topology management method in such networks. Grouping nodes are used as subroutines for executing various tasks in a distributed manner and for resource management, see the survey \cite{shahraki2020clustering} for an overview of different clustering methods for wireless sensor networks. 

Our work is motivated by the following question: \textsl{``Is it possible to exploit the properties of geometric intersection graphs for obtaining coresets whose size does not depend on the data size?''} In general,  the answer to this question is \emph{no}. 
This is because geometric intersection graphs can contain large cliques. Even for objects as simple as unit squares, the corresponding intersection graph could be a clique, and, as we already noted,  by setting suitable weights on the edges of the clique, one can represent any metric.
Hence constructing coresets in geometric intersection graphs with arbitrary edge weights is as difficult as in general metrics. Thus, we need to restrict edge weights in some manner in order to obtain non-trivial coresets for geometric intersection graphs. As an illustrative example, let us take a look at Euclidean-weighted UDGs, a well-studied class of geometric intersection graphs.

\medskip\noindent\textbf{Euclidean-weighted unit-disk graph metric.} A unit-disk graph (UDG) is defined in the following way---there is a configuration of closed disks of radii 1 in the plane and a one-to-one correspondence between the vertices and the centers of the disks such that there is an edge between two vertices if and only if the disks having the two corresponding centers intersect. The weight of an edge is equal to the Euclidean distance between the two corresponding centers. 
Euclidean-weighted UDGs have been well-studied in computational geometry \cite{ChanS19a,gao2005well}. Apart from practical motivation, UDGs are interesting from theoretical perspectives as well. On the one hand, being embedded on the plane they resemble planar graphs when ``zoomed out'', but could contain large cliques locally. On the other hand, the metric induced by them is an amalgamation of geometric and graphic settings, as it is locally Euclidean but globally a graph metric. Due to the latter property, UDG metric can be used for fine-tuned clustering, as with pure Euclidean distances one can only retrieve clusters induced by convex partitions of the space (see 
\Cref{fig:fig1}). 

\begin{figure}[h]
	\centering
	\begin{framed}
		\includegraphics[width=.9\textwidth]{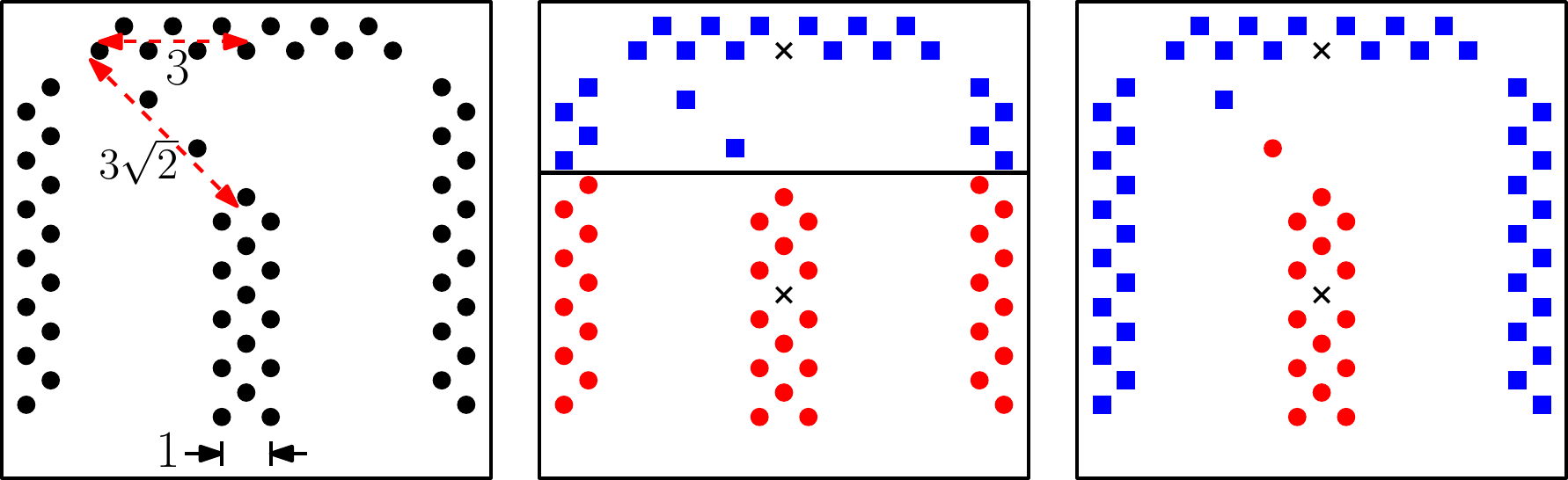}
		\caption{(Left.) A point set in 2D. (Middle.) $2$-means clustering with Euclidean distances. Centers are shown by crosses. Points of two clusters are shown by disks (red) and squares (blue). (Right.) $2$-means clustering on UDG -- due to the $3\sqrt{2}$ diagonal path distance, all the blue points (squares) are closer to the upper center.}
		\label{fig:fig1}
	\end{framed}
\end{figure}

\subsection{Our Results.}

We now formalize our intuition about the ``hybrid'' nature about the Euclidean weighted UDGs, by identifying two canonical geometric properties of a graph $G$ that are sufficient for constructing small-sized coresets. For better exposition,  we fix a few notations. For any subgraph $H$ of $G$, we denote its set of vertices and set of edges by $V(H)$ and $E(H)$, respectively. For   vertex set $V'\subseteq V$, we denote by  $G[V']$ the subgraph of $G$ induced by $V'$. For a subgraph $H$ of $G$, and   $u, v \in V(H)$, let $\pi_H(u, v)$ denote a shortest path between $u$ and $v$ (according to the edge-weights in $G$ restricted to $H$) that uses the edges of $H$, and let $\d_H(u, v)$ denote the weight of $\pi_H(u, v)$, i.e., the sum of the weights of the edges along $\pi_H(u, v)$. For any path $\pi$ in a graph, let $|\pi|$ denote the number of edges on $\pi$. Note that $\d_H$ is the so-called \emph{shortest path metric} on $H$. Finally, for any pair of points $p, q \in \real^2$, let $|pq|$ denote the euclidean (i.e., $\ell_2$-norm) distance between $p$ and $q$.

\medskip\noindent\textbf{Canonical geometric properties.}\\ 
\medskip\noindent(1) \textsl{Locally Euclidean:} There exist (not necessarily distinct) constants $c_1, c_2, c_3, c_4 \ge 0$, such that the following holds. $G$ has an embedding $\lambda: V(G) \to \real^2$ in the plane such that the vertices of $G$ are mapped to points in the plane, with the following two properties.
\begin{enumerate}
	\item For any two $u, v \in V(G)$, if $|\lambda(u)\lambda(v)| \le c_1$ then $uv \in E(G)$, and for any $u', v' \in V(G)$, if $|\lambda(u') \lambda(v')|_2 > c_2$, then $u'v' \not\in E(G)$. 
	\item For any $u, v \in V(G)$ such that $uv \in E(G)$, let $w(uv)$ denote the weight of the edge $uv$. Then, the edge $uv$ is a shortest path between $u$ and $v$ in $G$. 
	\\Furthermore, $c_3 \cdot |\lambda(u)\lambda(v)| \le w(uv) \le c_4 \cdot |\lambda(u) \lambda(v)|$. 
\end{enumerate}
\medskip\noindent(2) \textsl{Planar Spanner:} For any induced subgraph $G' = G[V']$ with $V' \subseteq V(G)$, there exists a planar $\alpha$-spanner $H' = (V', E(H'))$ for some fixed $\alpha \ge 1$, i.e., (i) $H'$ is a subgraph of $G'$ (and hence of $G$) -- $E(H') \subseteq E(G')$, and (ii) for any $u, v \in V'$, $\d_{G'}(u, v) \le \d_{H'}(u, v) \le \alpha \cdot \d_{G'}(u, v)$. 



Our main result is the following theorem. 
\begin{theorem}[Informal]
	\label{thm:coreset}
	Consider the metric space $(V,\d_G)$ induced by any graph $G$ satisfying the two canonical geometric properties (1) and (2), and a set $P\subseteq V(G)$. Then there exists a polynomial time algorithm that constructs a coreset for $(k,z)$-clustering on $P$ of size $\Oh(\epsilon^{-\beta}k\log^2k )$, where  $\beta = \Oh(z \log z)$. 
\end{theorem}


\Cref{thm:coreset} is a handy tool to construct coresets for several interesting geometric intersection graphs coupled with suitable metrics. First, let us observe that our initial example, namely, a metric induced by a Euclidean-weighted UDG $G$ satisfies the two canonical properties. Consider an embedding of $G$ in $\real^2$. Note that there is an edge between any two points iff the Euclidean distance between the two points is at most $2$, and the weight of such an edge is exactly the euclidean distance. Thus, $G$ is \textsl{Locally Euclidean} with $c_1 = c_2 = 2$, and $c_3 = c_4 = 1$. Furthermore, due to a result of Li, Calinescu, and Wan \cite{li2002distributed}, any Euclidean-weighted UDG admits a constant-stretch planar spanner (cf. \Cref{prop:udg-spanner}). Thus, $G$ also satisfies the \textsl{Planar Spanner} property. Therefore, due to \Cref{thm:coreset}, we can obtain $\Oh(\epsilon^{-\beta} k \log^2 k)$-size coresets for $(k, z)$-clustering on Euclidean-weighted UDGs. In the following, we discuss further applications of our framework.

\medskip
\noindent
\textbf{$\ell_\infty$-weighted unit-square graph metric.} Unit-square graphs (USGs) are similar to UDGs except they are defined as intersection graphs of (axis-parallel) unit squares instead of unit disks \footnote{Although it might seem unnatural at first, it is convenient to define a \emph{unit square} as a square of sidelength $2$. This is analogous to a unit disk being a disk of diameter $2$. In either case, the class of USGs remains unaffected by scaling.}. Indeed, these two graph classes are distinct. For example, the $K_{1,5}$ claw can be realized by a UDG, but not by any USG. (See  \Cref{fig:fig2}.) Since a unit square is a ``unit ball'' in $\ell_\infty$-norm, it is more natural to consider $\ell_\infty$ weights on the edges. It is not too difficult to see that the \textsl{Locally Euclidean} property holds for $\ell_\infty$-weighted USGs -- we give a formal proof in \Cref{subsec:usg}. On the other hand, in order to establish the second property, we have to prove the existence of a constant-stretch planar spanner for USGs. To the best of our knowledge this result was previously not known and is of independent interest. We show this result in \Cref{sec:usg-spanner}. Thus, $\ell_\infty$ weighted USGs also satisfy the two properties required to apply \Cref{thm:coreset} in order to obtain a small-sized coreset.

\begin{figure}[h]
	\centering
	\begin{framed}
		\includegraphics[width=.6\textwidth]{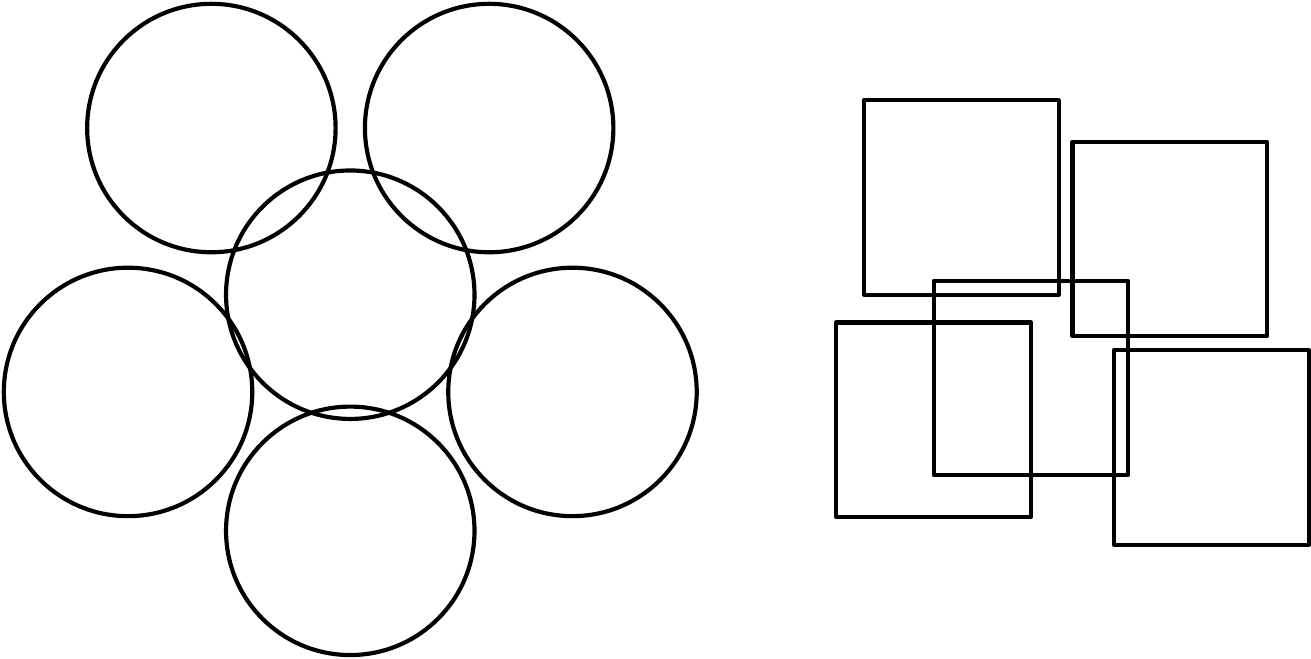}
		\caption{Figure showing a set of disks realizing $K_{1,5}$ and a set of squares realizing $K_{1,4}$. Two pairs of intersecting unit squares must contain a corner of each other, and so the central square can intersect with at most four other unit squares that are pairwise disjoint. }
		\label{fig:fig2}
	\end{framed}
\end{figure}
\medskip
\noindent
\textbf{Other extensions.} In $\real^2$, all $\ell_p$ distances ($1 \le p \le \infty$) are within a $\sqrt{2}$ factor from each other. Thus, our arguments can be easily extended to any $\ell_p$ weights on UDGs/USGs for any $1 \le p \le \infty$ without any changes on the bounds. We formally prove this in \Cref{subsec:other-norms}. Lastly, we consider shortest-path metrics in unweighted (i.e., hop-distance) unit-disk graphs of bounded maximum degree. Notably, these graphs satisfy the \textsl{Planar Spanner} property due to a result of \cite{Biniaz20}. Nevertheless, we show in \Cref{subsec:hop-udg} that we can modify our approach to construct a small-sized coreset for such metrics. To summarize, we obtain coresets for $(k, z)$-clustering with size independent of $n$ for the following graph metrics.
\begin{itemize}
	\item $\ell_p$-distance weighted UDGs for any $1 \le p \le \infty$,
	\item $\ell_p$-distance weighted USGs for any $1 \le p \le \infty$,
	\item Bounded-degree unweighted UDGs.
\end{itemize}

\medskip\noindent\textbf{FPT Approximation Schemes.} As a corollary to \Cref{thm:coreset}, we obtain $(1+\epsilon)$-approximations for $(k,z)$-clustering in geometric intersection graphs that are fixed-parameter tractable (FPT) in $k$ and $\epsilon$. Note that such a $(1+\epsilon)$-approximation was not known before even for UDGs, as it does not follow from previously known bound on coreset sizes. Prior to our work, the best known bound for UDGs --- as in general metrics --- was $O(k\log n\cdot \epsilon^{-\max(2,z)})$ \cite{Cohen-AddadSS21}. We note that even though a coreset reduces the number of distinct points (or clients) to be clustered, the number of potential centers (or facilities) still remains the same, i.e., $n$. Hence, a coreset does not directly help us enumerate all possible sets of $k$ centers from which we could pick the best set. An alternative way to enumerate these sets of centers is to enumerate all possible partitions (or clusterings) of the coreset points. Note that each clustering of coreset points corresponds to a clustering of the original points, and the cost of clustering is preserved to within a $(1\pm \epsilon)$ factor.
With our coreset bound of $\Oh(\epsilon^{-\beta}k\log^2k)$, the number of distinct clusterings is only $k^{\Oh(\epsilon^{-\beta}k\log^2k)}$. As the time complexity is dominated by the computations of cluster centers and costs for all the partitions, overall the algorithm takes $k^{\Oh(\epsilon^{-\beta}k\log^2k)}n^{\Oh(1)}$ time. 

\begin{corollary}
	For each of the metrics listed in the above, there exists a $(1+\epsilon)$-approximation for $(k,z)$-clustering with $z\ge 1$ that runs in time $2^{\Oh(\epsilon^{-\beta}k\log^3k)}n^{\Oh(1)}$, where $\beta = \Oh(z \log z)$. 
\end{corollary}

\subsection{General overview of the methods}

Our coreset construction is based on a  recent work due to Cohen-Addad, Saulpic, and Schwiegelshohn \cite{Cohen-AddadSS21}, which gives a   framework for constructing coresets in various settings.The essence of the framework is that it translates the problem of coreset construction to showing the existence of an interesting set of centers or centroid set. In particular, consider any set $\calS$ of $k$ centers and any subset $X\subseteq P$ of points that are sufficiently close to $\calS$ compared to an existing solution $\A$. Then a subset $\bbC\subseteq \Omega$ is a \emph{centroid set} for $P$ if it contains centers that well-approximates $\calS$, i.e., there exists $\tilS\subseteq \bbC$, such that for every $p\in X$, it holds that $|\cost(p, \calS)-\cost(p, \tilS)| \le {\epsilon}(\cost(p, \mathcal{S}) + \cost (p, \A))$. The framework informally states that if there is a centroid set $\bbC$, then a coreset can be constructed whose size depends logarithmically on $|\bbC|$. Such a dependency arises in their randomized construction in order to prove a union bound over all possible interesting solutions, which can be at most $|\bbC|^k$. By showing the existence of small-sized centroid sets, they obtain improved coreset size bounds for a wide range of spaces. 

We use the   framework of Cohen-Addad  et al. for our coreset construction. Our main technical result shows if for a graph $G$ with metric $\d_G$, the two canonical geometric properties are satisfied, then there exists a small-sized centroid set for $G$. This is the most challenging part of the proof and it requires a novel combination of tools and techniques from computational geometry. As soon as we establish the existence of  the centroid set, the construction of coresets follows the steps of \cite{Cohen-AddadSS21}. For the sake of exposition, let us consider a concrete example of Euclidean-weighted UDGs.

The first hurdle one faces while dealing with UDGs is that they encapsulate a combination of the Euclidean case and the case of \emph{graphical} or shortest-path metric. For example, consider any cluster of points with cluster center $s$. The points that are nearby (i.e., within distance $2$) $s$ behave simply as points in the Euclidean case. But, a point $p$ that is far away from $s$ can have a shortest path distance which is much larger than the actual Euclidean distance between $p$ and $s$, see \Cref{fig:fig1}. We first show that it is possible to conceptually separate out these two cases---but one has to be careful, as a cluster can potentially contain both types of points. Notably, none of the previous works had to deal with such a hybrid metric. To handle the set of nearby points, we exploit the \textsl{Locally Euclidean} property. In particular, by overlaying a grid of appropriately small sidelength, and selecting one representative point from each cell of the grid, we can compute a centroid set that preserves the distances from the nearby points. 

In the other case, a shortest path between a point $p$ and a center $s$ consists of more than one edge, and we need to deal with a graphical metric. This case is much more interesting. All other works establishing small-sized centroid sets in certain graph metrics exploit the fact that certain graph classes admit small or well-behaved separators. For example, bounded treewidth graph admit separators bounded by treewidth; whereas graphs excluding a fixed minor admit shortest path separators. However, UDGs may contain arbitrarily large cliques, and therefore do not admit such separators in general. Thus, we reach a technical bottleneck. Note that this is the first work of its kind that handles such a dense graph. To overcome this challenge, we use the other canonical property. Instead of directly working with the UDG, we consider its planar spanner, where distances are preserved up to a constant factor. The existence of such a spanner is guaranteed by the second canonical property, \textsl{Planar Spanner}. As planar graphs have shortest path separators, now we can apply the existing techniques. However, if we were to entirely rely on the spanner, some of the distances may be scaled up by a constant ($> 2$) factor in the spanner, and thus it would not be possible to ensure the $(1\pm \epsilon)$-factor bound required to construct a coreset. 

Thus, we use the spanner as a supporting graph in the following way. First, we recursively decompose the original UDG by making use of the shortest path separators admitted by the planar spanner. We note that although planar graph decomposition has been used in coreset literature, using such a guided scheme to obtain a decomposition of a much more general graph is novel. Then, we use this recursive decomposition of the UDG, along with the shortest path separators used to find this decomposition, in order to construct the centroid set. In this construction, we use the spanner in a restricted manner, and use it such that error incurred by the use of the spanner is upper bounded by $\alpha$ times the weight of at most one edge along a shortest path from a point $p$ to its corresponding (approximate) center. However, observe that if such a shortest path consists of a single edge, then even this error is too large. To resolve this issue, we rely on the planar spanner, only if the shortest path is ``long enough'', i.e., contains $\Omega(z/\epsilon)$ edges. In this case, \textsl{Locally Euclidean} property implies that for such a ``long path'', the \emph{length} of the path and the \emph{number of edges} on the path are within a constant factor from each other. This implies that the error introduced by rerouting a single edge using the spanner is at most $\epsilon$ times the length of the path, i.e., negligible.

Finally, if a shortest path between a point and a center consists of $\Oh(z/\epsilon)$ edges, then we can use a modified version of the grid-cell argument to obtain a small-sized centroid set.

We note that this is simply an intuitive overview of the challenges faced in each of the three cases. The actual construction of the centroid set, and the analysis of the error incurred in each of the cases is fairly convoluted. While replacing a center $s\in \calS$ by another one $\tils \in \bbC$, we need to ensure that for a point $p$ having $s$ as its closest center, $\d_G(p, \tils)$ is neither too large nor too small compared to $\d_G(p, s)$, since we want to bound the error in the absolute difference. In addition, we have to be extremely careful while combining the three centroid (sub)sets constructed for each of the cases, and ensure that a good replacement $\tils$ found for a center $s$ in one of the cases does not adversely distort the error for a point that is being handled in another case.

\medskip\noindent\textbf{Related work.} Here we give an overview of the literature on  
coresets.  For a more exhaustive list, we refer to  \cite{Cohen-AddadSS21,DBLP:conf/stoc/HuangV20}. 
Coreset construction was popularized by an initial set of works that obtained small-sized coresets in low-dimensional Euclidean spaces \cite{DBLP:conf/stoc/Har-PeledM04,DBLP:journals/dcg/Har-PeledK07,frahling2005coresets}. Chen \cite{chen2009coresets} obtained the first coreset for Euclidean spaces with polynomial dependence on the dimension and the first coreset in  general metrics, where the size is $\Oh(k^2\epsilon^{-2}\log^2 n)$ for $k$-median. Subsequently, the dependence on the dimension has been further improved \cite{langberg2010universal,feldman2011unified}. Finally, such dependence was removed in   \cite{feldman2013turning,sohler2018strong}. See also \cite{becchetti2019oblivious,Cohen-AddadSS21,DBLP:conf/stoc/HuangV20,braverman2021coresets,Cohen-AddadLSS22} for recent improvements. 

Both $k$-median and $k$-means admit polynomial-time $\Oh(1)$-approximations in general metrics \cite{CharikarGTS-STOC99,DBLP:conf/icalp/CharikarL12,DBLP:journals/jacm/JainV01,DBLP:journals/siamcomp/LiS16,DBLP:journals/talg/ByrkaPRST17,DBLP:conf/focs/AhmadianNSW17}. Moreover, algorithms with improved approximation guarantees can be obtained that is FPT in $k$ and $\epsilon$
\cite{DBLP:conf/icalp/Cohen-AddadG0LL19}. 
Naturally, the problems have also been studied in specialized metrics. 
Polynomial-time approximation schemes (PTASes) are known for Euclidean $k$-median~\cite{Arora} and $k$-means~\cite{DBLP:journals/siamcomp/Cohen-AddadKM19,DBLP:journals/siamcomp/FriggstadRS19}. See \cite{cohen2021near,DBLP:journals/jacm/KumarSS10} for other improvements. 
Similar to geometric clustering, clustering in graphic setting is also widely studied.
PTASes are known for excluded-minor graphs~\cite{DBLP:journals/siamcomp/Cohen-AddadKM19,braverman2021coresets}. 
Also, FPT approximation schemes are known for graphs of bounded-treewidth~\cite{DBLP:conf/icml/BakerBHJK020} and graphs of bounded highway dimension \cite{becker2018polynomial,braverman2021coresets}.

\ 


\section{Coresets for Geometric Graphs}
To set up the stage, we need the following definition of \textit{centroid set} from \cite{Cohen-AddadSS21}. 

\begin{definition}[Centroid Set]
Consider any metric space $(\Omega,\d)$, a set of clients $P\subseteq \Omega$, and two positive integers $k$ and $z$. Let $\epsilon >0$ be a precision parameter. Given a set of centers $\A$, a set $\bbC$ is an $\A$-approximate centroid set for $(k,z)$-clustering on $P$ that satisfies the following property. 

For every set of $k$ centers $\calS\subseteq \Omega$, there exists $\tilS\subseteq \bbC$, such that for every $p\in P$ that satisfies  $\cost(p, \calS) \le \lr{\frac{10z}{\epsilon}}^z \cdot \cost(p, \A)$ or $\cost(p, \tilS) \le \lr{\frac{10z}{\epsilon}}^z \cdot \cost(p, \A)$, it holds that
\[|\cost(p, \calS)-\cost(p, \tilS)| \le \frac{\epsilon}{z\log (z/\epsilon)}(\cost(p, \mathcal{S}) + \cost (p, \A)).\]
\end{definition}

Informally, a centroid set $\bbC$ is a collection of candidate centers, potentially much smaller than $\Omega$, such that the $k$ centers in $\calS$ can be replaced by $k$ centers in $\tilS\subseteq\bbC$ without changing the cost of points by a large amount, that are much closer to $\calS$ or $\tilS$ compared to $\A$ w.r.t. $\d$. They proved that one can obtain coresets whose size depends only logarithmically on the size of any such centroid set.    

\begin{proposition}[\cite{Cohen-AddadSS21}]
\label{prop:coreset} Consider any metric space $(\Omega,\d)$, a set of points $P\subseteq \Omega$ with $n$ distinct points, and two positive integers $k$ and $z$. Let $\epsilon >0$ be a precision parameter. Suppose $\A$ be a given constant-factor approximation for $(k,z)$-clustering on $P$. 

Suppose there exists an $\A$-approximate centroid set for $(k,z)$-clustering on $P$. Then there exists a polynomial time algorithm that constructs with probability at least $1-\delta$ a coreset of size 
\[\Oh\bigg(\frac{2^{\Oh(z\log z)}\cdot \log^4 (1/\epsilon)}{\min\{\epsilon^2,\epsilon^z\}}(k\log |\bbC|+\log \log (1/\epsilon)+\log (1/\delta))\bigg)\]
with positive weights for $(k,z)$-clustering on $P$.
\end{proposition}

First, note that the above coreset framework requires only the existence of such a centroid set. It is not needed to explicitly compute it. Indeed, such a centroid set is only used to bound the size of computed coresets in their analysis. 

Our main technical contribution is the theorem that guarantees the existence of a  {small-sized} centroid set when the geometric graph metrics satisfy two canonical geometric properties.  

\begin{theorem}
	\label{thm:centroidset} (Centroid Set Theorem) Consider the metric space $(V,\d_G)$ induced by any graph $G=(V,E)$ satisfying the \textsl{Locally Euclidean} and \textsl{Planar Spanner} properties defined before. Also consider a set of points $X\subseteq V$ and two positive integers $k$ and $z\ge 1$. Let $\epsilon >0$ be the precision parameter. Additionally, suppose $\A$ be a solution for $(k,z)$-clustering on $X$. Then there exists an $\A$-approximate centroid set $\mathbb{C}$ for $(k,z)$-clustering on $X$ of size $\exp (\Oh (\log^2 |X| + z^{16} \epsilon^{-8} (\log(z/\epsilon))^{8} \log|X| ))$. 
\end{theorem}
We prove Theorem \ref{thm:centroidset} in the following section. Then the desired coreset result follows by Proposition \ref{prop:coreset} and from \cite{braverman2021coresets}, with some minor changes due to our different bound on coreset-size. For completeness, we describe the proof. 



\begin{restatable}{theorem}{coreset}
	\label{cor:coreset}
	Consider the metric space $(V,\d_G)$ induced by any graph $G=(V,E)$ satisfying \textsl{Locally Euclidean} and \textsl{Planar Spanner} properties, a set $P\subseteq V$ with $n$ distinct points, and two positive integers $k$ and $z\ge 1$. Then there exists a polynomial time algorithm that constructs with probability at least $1-\delta$ a coreset for $(k,z)$-clustering on $P$ of size $\Oh(\epsilon^{-\Oh(z\log z)}k\log^2k \log^3 (1/\delta))$, where $z$ is a constant, and $\delta < 1/4$.  
\end{restatable}
\begin{proof}
	For proving the theorem, we apply the Iterative size reduction algorithm \cite{braverman2021coresets}. Let $t$ be the largest integer such that $\log^{(t-1)} n\ge \max\{125k\epsilon^{-\rho}\log (1/\delta), \rho 2^{\rho+1}\}$ for a sufficiently large constant $\rho$ to be set later. Also, let $X_0=P$. For $1\le i\le t$, set $\epsilon_i=\epsilon/(\log^{(i)} n)^{1/\rho}$  and $\delta_i=\delta/|X_{i-1}|$. Let $X_i$ be the coreset computed by applying the algorithm of  \Cref{prop:coreset} on the client set $X_{i-1}$ with $\epsilon$ being $\epsilon_i$ and $\delta$ being $\delta_i$ for $1\le i\le t$. Finally, let $U$ be the coreset returned by applying the same algorithm on $X_t$ with approximation guarantee $1+\epsilon$ and failure probability $\delta$. 

First we argue about the rate of decrement in coreset sizes between two consecutive iterations. Let $n_i$ be the size of $X_i$ for $1\le i\le t$. Then,

%
	\begin{align*}
		n_i & = \frac{2^{\Oh(z\log z)}\cdot\log^4 (1/\epsilon_i)}{\min\{\epsilon_i^2,\epsilon_i^z\}}(k\log |\bbC_i|+\log \log (1/\epsilon_i)+\log (1/\delta_i)) \tag{$\bbC_i$ is the centroid set w.r.t. $X_i$}\\
		& = \frac{2^{\Oh(z\log z)}\cdot\log^4 (1/\epsilon_i)}{\min\{\epsilon_i^2,\epsilon_i^z\}}(k\log^2 n_{i-1}+kz^9 \epsilon_i^{-5}  \log n_{i-1}+\log \log (1/\epsilon_i)+\log (1/\delta_i))\\
		& \le \epsilon_i^{-\alpha} k\log^2 n_{i-1} \log (1/\delta_i) \tag{$\alpha=\Oh(z\log z)$}\\
		& \le \epsilon_i^{-\alpha} k\log^2 n_{i-1} (\log n_{i-1}+\log (1/\delta))\\
		& \le \epsilon_i^{-\alpha} k\log^3 n_{i-1} \log (1/\delta)
	\end{align*}
	
	The last inequality follows, as $n_{i-1} \ge 4$ and $\delta < 1/4$. Let $\Gamma=k\epsilon^{-\alpha}\log (1/\delta)$. Next, we prove that $n_i\le 125\Gamma (\log^{(i)} n)^4$ for $1\le i\le t$. We use induction on $i$. In the base case, $n_1\le \epsilon_1^{-\alpha} k\log^3 n \log (1/\delta)=k\epsilon^{-\alpha}\log (1/\delta) (\log n)^{\alpha/\rho} \log^3 n =\Gamma \log^4 n\le 125 \Gamma \log^4 n$. The second last equality follows by setting the value of $\rho$ to be equal to $\alpha$. Now, consider any $i \ge 2$. 
	
	\begin{align*}
		n_i &\le \epsilon_i^{-\alpha} k\log^3 n_{i-1} \log (1/\delta)\\
		& \le (\epsilon/(\log^{(i)} n)^{1/\rho})^{-\alpha} k\log^3 n_{i-1} \log (1/\delta)\\
		& = k\epsilon^{-\alpha}\log (1/\delta) (\log^{(i)} n)^{\alpha/\rho} \log^3 n_{i-1}\\
		& = \Gamma (\log^{(i)} n) \log^3 n_{i-1} \tag{as $\rho = \alpha$}\\
		& \le \Gamma (\log^{(i)} n) (\log (125\Gamma (\log^{(i-1)} n)^4))^3\\
		& \le \Gamma (\log^{(i)} n) (\log (125\Gamma) + 4\log^{(i)} n)^3\\
		& \le \Gamma (\log^{(i)} n) (\log^{(i)} n+4\log^{(i)} n)^3 \tag{by definition of $t$, $\log^{(i-1)} n\ge 125k\epsilon^{-\rho}\log (1/\delta)=125\Gamma$}\\
		& = 125\Gamma (\log^{(i)} n)^4
	\end{align*}

It follows that $n_t \le 125\Gamma (\log^{(t)} n)^4$. Now, by definition of $t$, $\log^{(t)} n < \max\{125\Gamma, \rho 2^{\rho+1}\}$. Hence, $n_t \le \max\{(125\Gamma)^5,125\Gamma (\rho 2^{\rho+1})^4\}$. Next, we analyze the error in the following claim.
\begin{restatable}{claim}{epsproof} \label{cl:epsproof}
	$\prod_{i=1}^t (1+\epsilon_i) \le 1+10\epsilon$.
\end{restatable}

\begin{proof}
	By definition, $X_t$ is a coreset with approximation guarantee at most $\Pi_{i=1}^t (1+\epsilon_i)$. Now, by our assumption, $\log^{(i-1)} n\ge \rho 2^{\rho+1}$ for $i\le t$. For such value of $\log^{(i-1)} n$, 
	
	\begin{align*}
		& \log^{(i-1)} n\ge 2^{\rho} \log^{(i)} n\\
		& \Rightarrow (\log^{(i-1)} n)^{1/\rho} \ge 2 (\log^{(i)} n)^{1/\rho}\\
		& \Rightarrow \epsilon/(\log^{(i-1)} n)^{1/\rho}\le (1/2)\cdot (\epsilon/(\log^{(i)} n)^{1/\rho})\\
		& \Rightarrow \epsilon_i\ge 2\epsilon_{i-1}\\
		& \Rightarrow \sum_{i=1}^t \epsilon_i \le 2\epsilon_t
	\end{align*}
	
	It follows that \[\prod_{i=1}^t (1+\epsilon_i)\le \exp\lr{\sum_{i=1}^t \epsilon_i}\le \exp\lr{2\epsilon_t}\le \exp\lr{\frac{2\epsilon}{(\log^{(t)} n)^{1/\rho}}}\le \exp\lr{2\epsilon}\le 1+10\epsilon.\]
	The second last inequality follows from the fact that $\log^{(t)} n\ge 1$, which is true as $\log^{(t-1)} n\ge \rho2^{\rho+1}$.   
\end{proof}

Next, we analyze the failure probability. First note that we can assume that the size of $X_{i-1}$ is at least $\log^{(i-1)} n$. Otherwise, we can add arbitrary points to the coreset $X_{i-1}$ and increase its size. Hence, $n_{i-1}\ge \log^{(i-1)} n$. Then, $\delta_i=\delta/n_{i-1}\le \delta/\log^{(i-1)} n$. So, the total failure probability is at most, 

\[\sum_{i=1}^t \delta_i=\delta \lr{\frac{1}{n}+\frac{1}{\log n}+\ldots+\frac{1}{\log^{(t-1)} n}}=\Oh(\delta)\]
The last inequality follows, as $\log^{(t-1)} n\ge \rho2^{\rho+1}\ge 4$.     

In the above we showed that $X_t$ is a coreset of size at most $\max\{(125\Gamma)^5,125\Gamma (\rho 2^{\rho+1})^4\}$ with approximation guarantee $1+\Oh(\epsilon)$ and failure probability at most $\Oh(\delta)$. Hence, the set $U$ obtained by applying the  algorithm of  \Cref{prop:coreset} on $X_t$ is a coreset with approximation guarantee $1+\Oh(\epsilon)$, failure probability at most $\Oh(\delta)$ and size at most 

\begin{align*}
	\epsilon^{-\alpha} k\log^2 |X_t| \log (1/\delta) & = \Oh(\epsilon^{-\alpha}k\log (1/\delta)\log^2 \Gamma)\\
	&=\Oh(\epsilon^{-\alpha}k\log (1/\delta)(\log k+\log (1/\epsilon)+\log\log (1/\delta))^2)\\
	&=\Oh(\epsilon^{-\beta}k\log^2k \log^3 (1/\delta))
\end{align*}
where $\beta = \Oh(\alpha) = \Oh(z\log z)$. Scaling $\epsilon$ and $\delta$ by constants we obtain the desired bound. Finally, as $t+1=\Oh(\log^* n)$, we apply the $\Oh(n)$ time algorithm in  \Cref{prop:coreset} $\Oh(\log^* n)$ times. Moreover, we need to compute an approximate solution $\A$ in order to apply this algorithm. For this purpose, we use the $\tilde{\Oh}(nk)$ time algorithm of Mettu and Plaxton \cite{mettu2004optimal} that returns a constant approximation.
\end{proof}

\newcommand{\btrue}{\texttt{true}}
\newcommand{\bfalse}{\texttt{false}}

\section{Proof of the Centroid Set Theorem}

In this section, we prove our main result, existence of a small-sized centroid set. 
Recall that we are given $G = (V, E)$, a connected, undirected, and edge-weighted graph on $n$ vertices. Moreover, $G$ satisfies two canonical geometric properties: (1) \textsl{Locally Euclidean}, and (2) \textsl{Planar Spanner}.

As $(V,\d_G)$ is our metric space, we use the terms points and vertices interchangeably. $X\subseteq V$ is the given set of points. We are also given $\A$, a solution for $(k,z)$-clustering on $X$. We prove that there exists an $\A$-approximate centroid set of size $\exp (\Oh (\log^2 |X| + z^{16} \epsilon^{-8} (\log(z/\epsilon))^{8} \log|X|))$ for $(k,z)$-clustering on $X$, which satisfies the following property. 

For every set of $k$ centers $\calS\subseteq V$, there exists $\tilS\subseteq \bbC$, such that for every $p\in X$ that satisfies  $\cost(p, \calS) \le \lr{\frac{10z}{\epsilon}}^z \cdot \cost(p, \A)$ or $\cost(p, \tilS) \le \lr{\frac{10z}{\epsilon}}^z \cdot \cost(p, \A)$, it holds that 
\[|\cost(p, \calS)-\cost(p, \tilS)| \le \frac{\epsilon}{z\log (z/\epsilon)}(\cost(p, \mathcal{S}) + \cost (p, \A)).\]

\subsection{Useful Properties}

In the analysis, we will use a few consequences of the \textsl{Locally Euclidean} property, which we state in the following.

\begin{proposition}[Bounded Distance.] \label{prop:bounded-distance}
	There exist universal constants $c'_1, c'_2$ such that, the weight of any edge of $G$ is at most $c'_2$, and for any $\tau \ge 2$, and for any $u, v \in V(G)$, if $|\pi_G(u, v)| = \tau$, then $\d_G(u, v)\ge c'_1 \tau$. That is, $\d_G(u, v) = \Theta(\tau)$ for $\tau \ge 2$. 
\end{proposition}
\begin{proof}
	Let $c_1, c_2, c_3, c_4 \ge 0$ be the constants such that $G$ and the embedding $\lambda$ satisfies the \textsl{Locally Euclidean} property. Consider an edge $uv \in E(G)$. It follows that $|\lambda(u) \lambda(v)| \le c_1$, which implies that $w(uv) \le c_4 \cdot |\lambda(u) \lambda(v)| \le c_4 \cdot c_1$. Hence, the weight of any edge is bounded by $c'_2 \coloneqq c_1 \cdot c_4$.
	
	Now, consider a minimum-hop shortest path $\pi(u, v) = (u = w_0, w_1, \ldots, w_{\tau} = v)$, where $\tau \ge 1$. For any $0 \le i \le t-2$, observe that $w_i w_{i+2} \not\in E(G)$ -- otherwise we can short-cut and obtain a path of smaller length as well as smaller number of hops. This implies that, $|\lambda(w_{i}) \lambda(w_{i+2})| > c_1$. Then, by triangle inequality, it follows that $ |\lambda(w_{i}) \lambda(w_{i+1})| +  |\lambda(w_{i+1}) \lambda(w_{i+2})| > c_1$.
	On the other hand, $\d_G(w_{i}, w_{i+2}) = \d_G(w_i, w_{i+1}) + \d_{G}(w_{i+1}, w_{i+2}) \ge c_3 \cdot ( |\lambda(w_{i}) \lambda(w_{i+1})| +  |\lambda(w_{i+1}) \lambda(w_{i+2})| ) > c_3 \cdot c_1$. That is, the length of each sub-path of $\pi(u, v)$ of length $2$ is at least a constant. It follows that, $\d_G(u, v) \ge c_3 \cdot c_1 \cdot \lfloor \tau/2 \rfloor \ge c'_1 \cdot \tau$ if $\tau \ge 2$.  
\end{proof}

\begin{proposition}[Construction of $\mu$-net] \label{prop:mu-net}
	Let $0 < \mu < c_1/\sqrt{2}$, and $t \ge 0$. Then, for any $v \in V(G)$, let $B(v, r) = \{ u \in V(G) : \d_G(u, v) \le r \}$. Then, there exists a subset $B' \subseteq B(v, r)$ of size $\Oh(r^2/\mu^2)$ such that for any $u \in B(v, r)$, there exists $u' \in B'$ such that $\d_G(u, u') = \Oh(\mu)$. Such a set $B'$ is called a $\mu$-net of $B(v, r)$.
\end{proposition}
\begin{proof}
	Consider an embedding $\lambda: V(G) \to \real^2$ guaranteed by \textsl{Locally Euclidean} property. Consider a grid of sidelength $\mu < c_1 \sqrt{2}$. It follows that for any two points contained in a grid cell, the corresponding vertices are neighbors in $G$. Furthermore, for any neighbors $u, v \in V(G)$, it follows that $|\lambda(v) \lambda(f(v))| \le \sqrt{2} \mu \le c_1$, which implies that $\d_G(v, f(v)) = w(v f(v)) \le c_4 \cdot |\lambda(v) f(v)| \le c_4 \cdot \sqrt{2} \mu = \Oh(\mu)$. To construct $B'$, we arbitrarily select one point from each cell $C \cap B(v, r)$, provided that $C \cap B(v, r) \neq \emptyset$. 
	
	To bound the size of $B'$, we consider two cases. If $r \le c_1$, then all vertices in $B(v, r)$ are neighbors of $v$. It follows that all the corresponding points are contained in a disk of radius $r/c_3$ around $\lambda(v)$. The number of grid cells of sidelength $\mu$ contained in this disk can be upper bounded by $\Oh(r^2/\mu^2)$, and we add one point from each such cell to $B'$. If $c_1 < r < c_2$, then we can again use a similar argument to upper bound the number of cells by $\Oh(c_2^2/\mu^2) = \Oh(r^2/\mu^2)$.
	
	Otherwise, suppose $r > c_2$. Consider a vertex $w$ such that $\d_G(v, w) \le r$, but $v$ and $w$ are not neighbors in $G$. It follows that $|\pi_G(v, w)| = \tau \ge 2$. It follows that $|\lambda(v) \lambda(w)| \le c_4 \tau$. On the other hand, \Cref{prop:bounded-distance} implies that $\d_G(v, w) \ge c'_1 \tau$. This implies that $|\lambda(v) \lambda(w)| \le (c_4/c'_1) \cdot r$. It follows that a point corresponding to a vertex in $B(v, r)$ is contained in a disk of radius $\Oh(r)$ around $\lambda(v)$. The number of cells of sidelength $\mu$ inside this disk is $\Oh(r^2/\mu^2)$.
\end{proof}

\begin{proposition}[Bounded Degree Support Graph] \label{prop:support-graph}
	For any \emph{precision value} $0 < \mu < c_1/\sqrt{2}$, there exists a graph $H$ with the following properties (i) $V(H)\subseteq V(G)$, the maximum vertex degree of $H$ is bounded by a polynomial function $\Delta(1/\mu) = \Oh(1/\mu^2)$, and (iii) there is a mapping $f:V(G)\rightarrow V(H)$ such that for any $v\in V(G)$, $\d_G(v,f(v))= \Oh(\epsilon)$ and for any path $u_1,u_2,\ldots,u_t$ of $G$, $H$ contains the path $f(u_1),f(u_2),\ldots,f(u_t)$.  
\end{proposition}
\begin{proof}
	Consider a gridcell of sidelength $\mu$ as in the previous proposition, which implies that any two points belonging to a grid cell are neighbors in $G$. Now, we construct support graph $H$ in the following way. We select any arbitrary point from each grid cell that contains a point in $V(G)$ and add this special point to $V(H)$. The mapping $f$ is constructed in the following way. Map each vertex $v\in V(G)$ to the special point of the grid cell that contains $v$. It follows that $|\lambda(v) \lambda(f(v))| \le \sqrt{2} \mu \le c_1$. It also follows that $\d_G(v, f(v)) = w(v f(v)) \le c_4 \cdot |\lambda(v) f(v)| \le c_4 \cdot \sqrt{2} \mu = \Oh(\mu)$.
	
	Next, we define the edges of $H$. Consider two gridcells $C_1$ and $C_2$ with their special points $p_1$ and $p_2$, respectively. We add the edge $p_1p_2$ to $H$ if $C_1$ and $C_2$ contain two points that are neighbors in $G$. Thus, the maximum distance between two special points that can contain an edge is at most $c_2+\Oh(\mu)$. As we pick at most one special point from each grid cell, it follows that the degree of any special point in $H$ is bounded by $\Oh(1/\mu^2)$. Now, consider any path $u_1,u_2,\ldots,u_t$ of $G$. Then the edge between two special points $f(u_i)$ and $f(u_{i+1})$ exists in $H$ for $1\le i\le t-1$. Hence, the path $f(u_1),f(u_2),\ldots,f(u_t)$ is in $H$.
\end{proof}



In the following analysis, we will often need the following generalization of triangle inequality that works for powers of distances.
\begin{lemma}[Triangle Inequality for Powers \cite{MakarychevMR19}] \label{lem:triangle}
	Let $\epsilon > 0$ and $z \ge 1$.
	\begin{enumerate}
		\item For any non-negative reals $x, y$, $ (x+y)^z \le (1+\epsilon)^{z-1} x^z + \lr{\frac{1+\epsilon}{\epsilon}}^{z-1} y^z.$
		\item Let $(\Omega, \d)$ be a metric space, and $u, v, w \in \Omega$, then,
		\[ \d(u, v)^z \le (1+\epsilon)^{z-1} \cdot (\d(u, v))^z + \lr{\frac{1+\epsilon}{\epsilon}}^{z-1} (\d(v, w))^z.  \]
	\end{enumerate}
\end{lemma}

To prove the centroid set theorem, we show an explicit construction of a centroid set. Recall that, due to the first canonical property, points in $V$ that are close to each other behaves as in the Euclidean case. To take care of the case of \textbf{points nearby} to their closest centers, we add a set of points $\bbCN$ to our centroid set. The case of far away points is further divided into two subcases. In the first subcase, we deal with the points whose shortest paths to closest centers are \textbf{short} or $O(1/\epsilon)$ hops away. To take care of this subcase,  
we add a set of points $\bbCS$ to our centroid set. The last subcase concerns \textbf{long paths}, and here we make use of the planar spanner property. In particular, we construct the centroid points in this case based on a recursive decomposition of the graph guided by underlying planar spanners of the decomposed subgraphs. This subcase resembles the centroid set construction in excluded-minor graph metrics from \cite{Cohen-AddadLSS22,braverman2021coresets}. 

Next, we proceed towards the details of our construction. But, before we describe our approach, we need to define some ingredients, mainly to make sense of the decomposition of a graph guided by planar spanners.   

\subsection{Shortest Path Separators and Recursive Decomposition} \label{subsec:decomp}

\begin{proposition}[\cite{KawarabayashiST13,ChanS19a}] \label{thm:sp-separator}
	Given a planar graph $H = (V, E)$, with non-negative weights on vertices, there exists a collection of shortest paths $\P = \{P_1, P_2, \ldots, P_b\}$ with $b = \Oh(1)$, such that the weight of the set of vertices belonging to any connected component in $H \setminus \bigcup_{P_i \in \P} V(P_i)$ has weight at most half of that of $V(H)$.
\end{proposition}

Let $X \subseteq V(G)$ be a set of vertices, and let $I_X: V(G) \to \{0, 1\}$ be the indicator function for $X$. We obtain a recursive decomposition $\T = (\R, E(\T))$ as follows, where every node of the tree $R_i \in \R$ corresponds to a subset of $V(G)$, which is termed as a region.
\begin{itemize}
	\item The region $R_1$ corresponding to the root of $\T$ is equal to $V(G)$.
	\item Consider an internal node of $\T$ that corresponds to a region $R_i \subseteq V(G)$, such that $|R_i \cap X| > 2$. Then, let $G_i = G[R_i]$, and let $H_i$ be a planar $\alpha$-spanner for $G_i$. Then, we apply \Cref{thm:sp-separator} on $H_i$ to obtain a collection of $\Oh(1)$ \emph{separator paths} $\P_i$, such that each $P\ij \in \P_i$ is a shortest path in $H_i$ w.r.t.\ $I_X$ being the weight function (i.e., we seek to obtain a decomposition of the vertices of $X$ in a balanced manner). Then, the regions that become the children of $R_i$ in $\T$ are as follows. (See \Cref{fig:decomp} for an illustration)
	\begin{enumerate}
		\item The subsets of vertices corresponding to the connected components in the induced subgraph of $H_i$ on $R_i \setminus \bigcup_{P^i_j \in \P_i}V(P^i_j)$, and these are termed as \emph{component children} of $R_i$ and
		\item Each path in $\P_i$ is broken into maximal sub-paths each containing at most two points of $X$.\footnote{More precisely, if a path $P$ contains at least one vertex of $X$, then interpret as a real interval $[0, 1]$ with break-points $0 \le x_1 \le x_2 \le \ldots \le x_m \le 1$, where $x_i$'s are the vertices in $P \cap X$. Then, the sub-paths are given by $[0, x_1], [x_1, x_2], \ldots, [x_m, 1]$. If $P$ does not contain a vertex of $X$, then let $P$ itself be the sub-path.} Then, the vertex sets of each such sub-paths are added as different children of $R_i$. These are termed as \emph{subpath children}.
	\end{enumerate}
	Thus, a region $R_i$ is equal to the union of all the regions corresponding to its component and subpath children. Note that only the regions corresponding to the consecutive subpath children of a separator path may intersect at a vertex of $X$; whereas the regions corresponding to the component children of $R_i$ define a partition of $R_i \setminus \bigcup_{P\ij \in \P_i} V(P^i_j)$. 
	\item Recursion stops when a region contains at most $2$ vertices of $X$, and such a region is called a leaf region/node. Since the weight falls by a factor of at least $1/2$ at every level, the height of $\T$ is at most $\Oh(\log |X|)$. 
\end{itemize}

\begin{figure}
	\centering
	\begin{framed}
		\includegraphics[scale=0.8]{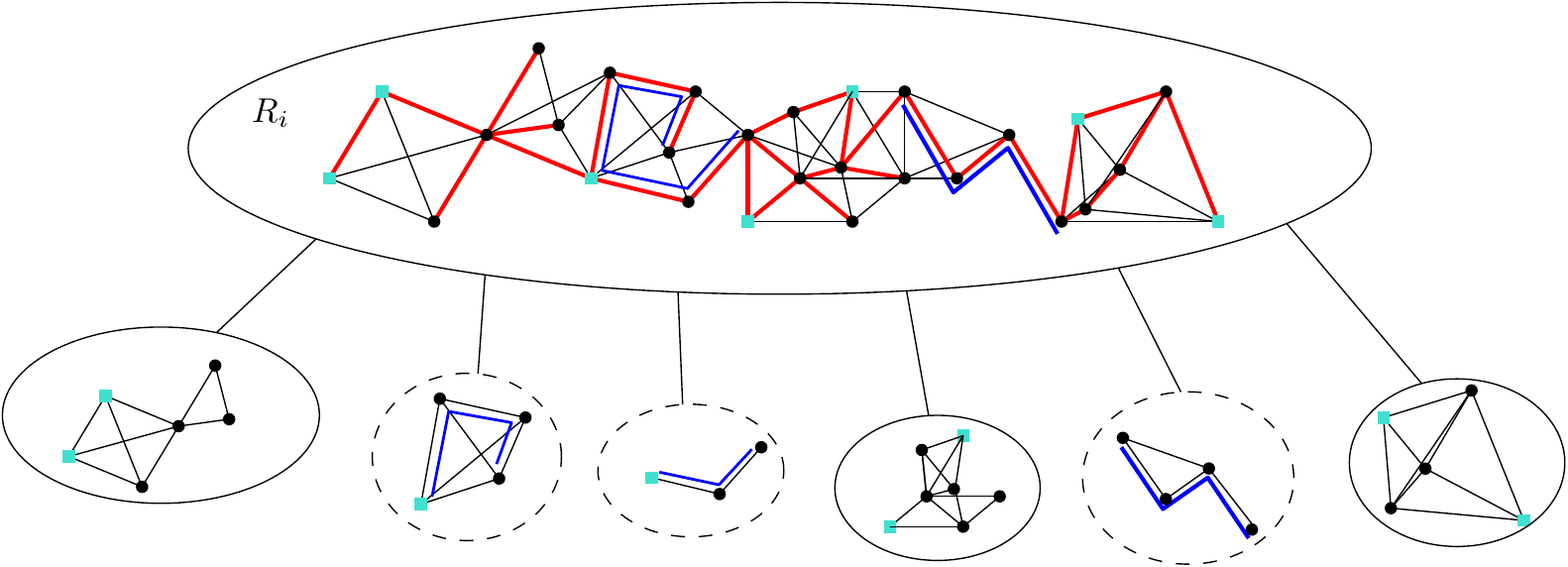}
		\caption{Decomposition of a region $R_i$ into multiple children. Inside $R_i$, we show the induced subgraph $G_i = G[R_i]$ (which is not planar) and the corresponding planar spanner $H_i$. The edges of $H_i$ are shown in red, and the edges in $E(G_i) \setminus E(H_i)$ are in black. The two shortest-path separators of $H_i$ are shown in blue, which form a balanced separator $\P_i$ for the vertices of $X$ (shown as light blue squares). Then, we add a child for every connected component of $H \setminus V(\P_i)$ (children regions inside solid ellipses). We also have children corresponding to paths in $\P_i$ (dashed ellipses). One of the paths is broken into two pieces due to a vertex in $X$.} \label{fig:decomp}
	\end{framed}
\end{figure}

Fix a vertex $s \in V(G)$, and let $R_1, R_2, \ldots, R_t$ be the root-leaf path in $\T$ such that $s$ is contained in every $R_i$, where $R_1$ is the root region, and $R_t$ is a leaf region (recall that $t = \Oh(\log |X|)$). Now consider any $p \in X \subseteq V(G)$. Let $R_i$ be the lowest (i.e., one with the maximum $i$) region that contains every vertex on a shortest path $\pi_G(p, s)$ in $G$ between $p$ and $s$. Additionally, suppose $R_i$ is not equal to $R_t$. Then, the construction implies that there exist two children $R_{i+1}, R'_{i+1}$ of $R_i$, such that $p \in R'_{i+1}$ and $s \in R_{i+1}$. Let $u$ be the \emph{last vertex} along $\pi_G(p, s)$ (while going from $p$ to $s$) that lies in $R'_{i+1}$, and $v$ be the vertex immediately after $u$ along $\pi_G(p, s)$. Note that $u, v$ belong to $R_i$, which induces a connected component in $H_i$, and thus in $G_i$. Let $\pi_{H_i}(u, v)$ be a shortest path between $u$ and $v$ in $H_i$. Since $u \in R'_{i+1}$, and $v \not\in R'_{i+1}$, it follows that $\pi_{H_i}(u, v)$ intersects with a path separator $P\ij \in \P_i$ at a vertex $x$, such that $x \in V(P\ij) \cap V(\pi_{H_i}(u, v))$. In this case, we say that $P\ij$ \emph{separates} $p$ and $s$, and $x$ is a separating vertex that lies between $u$ and $v$ (see  \Cref{fig:separation}). We summarize a few properties in the following discussion. 

\begin{figure}
	\centering
	\begin{framed}
		\includegraphics[scale=0.7]{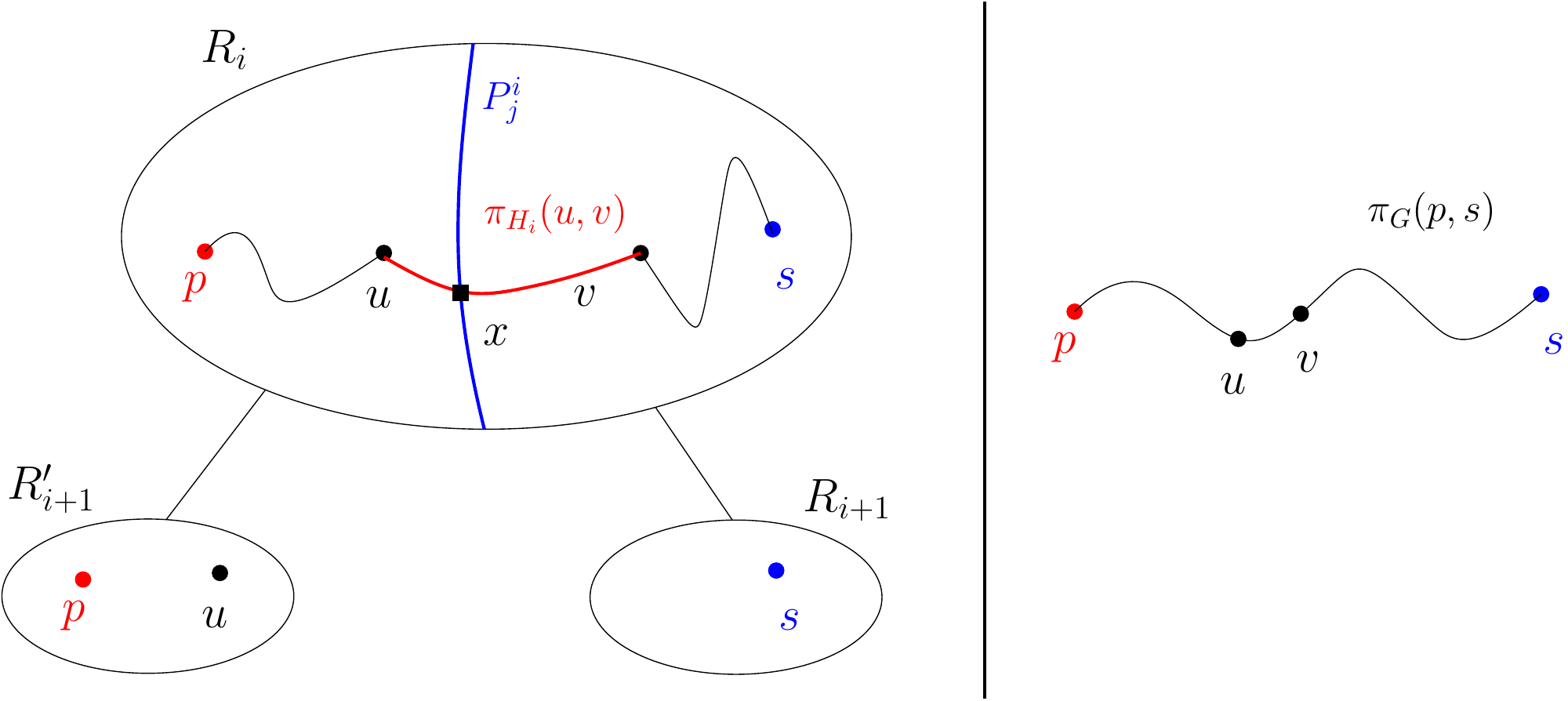}
		\caption{A shortest path $\pi_G(p, s)$ is depicted on the right. On the left hand side, we depict a path $P\ij$ separating $p$ and $s$ and $x$ is a separating vertex that lies between $u$ and $v$. Note that the paths ($P\ij$) inside $R_i$ are shortest paths in the planar spanner $H_i$. Also note that $u \in R'_{i+1}$, and $s \in R_{i+1}$ by definition, but $v$ may not belong to $R_{i+1}$.} \label{fig:separation}
	\end{framed}
\end{figure}

\begin{claim}\label{cl:separating-paths}
	Let $P\ij \in \P_i$ be a separator path corresponding to a region $R_i$ that is the lowest region containing a shortest path $\pi(s, p)$, such that $P\ij$ separates $p$ and $s$, and $x$ is a separating vertex that lies between $u$ and $v$. Then,
	\begin{enumerate}
		\item $\d_{G_i}(p, s) = \d_G(p, s)$,
		\item $\d_{G_i}(u, v) \le \d_{H_i}(u, v) \le \alpha \cdot \d_{G_i}(u, v) = \alpha \cdot \d_{G}(u, v)$, 
		\item $\d_{H_i}(u, x) + \d_{H_i}(v, x) = \d_{H_i}(u, v) \le \alpha \cdot \d_{G_i}(u, v) = \alpha \cdot \d_{G}(u, v)$
		\item For any $q > 0$, if $|\pi_G(p, s)| \ge q \cdot \frac{\alpha\cdot  c'_2}{c'_1\cdot \epsilon}$, then $\d_{H_i}(u, v) \le \epsilon/q \cdot \d_G(p, s) $.
	\end{enumerate}
\end{claim}
\begin{proof}
	The first property follows from the assumption that $R_i$ contains all the vertices along the shortest path $\pi_G(p, s)$, thus $G_i = G[R_i]$ also contains the shortest path $\pi_G(p, s)$. Now consider the second property. Note that $u$ and $v$ are consecutive vertices along $\pi_G(p, s)$, which is a shortest path in $G$ and $G_i$. Thus, there exists a shortest path $\pi_{H_i}(u, v)$ in $H_i$, which is an $\alpha$-spanner for $G_i$, of weight at most $\alpha$ times $\d_{G_i}(u, v)$. For the third inequality, recall that $x$ is a vertex along the path $\pi_{H_i}(u, v)$, and then we use the second item. 
	
	Since $\pi_G(p, s)$  contains at least $q \cdot \frac{\alpha c'_2}{\epsilon c'_1}$ vertices, the Bounded Distance property of $G$ implies that, $\d_G(p, s) \ge c'_1 \cdot q \frac{\alpha c'_2}{\epsilon c'_1} = q \cdot \frac{\alpha c'_2}{\epsilon}$, whereas $\d_G(u, v) \le c'_2$. Then, it follows that $\d_{H_i}(u, v) \le \alpha \cdot \d_G(u, v) \le \alpha c'_2\le \frac{\epsilon}{q} \cdot \d_G(p, s)$. 
\end{proof}
Note that the last item in the above claim infers that the distance between $u$ and $v$ in the spanner $H_i$ is negligible if $|\pi_G(p, s)|$ is large enough.  
\subsection{Construction of the Centroid Set}

We will construct the three sets of centers $\bbCL, \bbCS, $ and $\bbCN$, and define the centroid set as the union of the three sets, i.e., $\bbC \coloneqq \bbCL \cup \bbCS \cup \bbCN$. 
The three sets correspond to the three cases we discussed before. 

\subparagraph*{Nearby points case.} First, we construct the set $\bbCN$.
Let $p_1,p_2,\ldots,p_{n'}$ be the points of $X$ such that $\d_G(p,\A) < 1$ for all $p=p_i$, where $1\le i\le n'$. Let $B_i=B(p_i,(10z/\epsilon)\cdot \d_G(p_i,\A))$ be the set of points contained within distance $10z/\epsilon \cdot \d_G(p_i, \A)$ from $p_i$. Then, we use \Cref{prop:mu-net} to compute a set $B'_i \subseteq B_i$ by setting $\mu = \Theta(\epsilon^3/z^3) \cdot \d_G(p_i, \A)$, and add all points of every such $B'_i$ to $\bbCN$. 

\subparagraph*{Short-path case.} Next, we construct the set $\bbCS$. Consider the Bounded Degree Support Graph $H''=(V'',E'')$ obtained by applying \Cref{prop:support-graph} with precision value $\mu = \epsilon^2/z^2$. Also consider the mapping $f:V(G)\rightarrow V(H'')$. For each point $p\in X$, add any vertex $v\in V(H'')$ to $\bbCS$ that is at most $\ell=\frac{\gamma \cdot z \alpha\cdot  c'_2}{c'_1\cdot \epsilon}$ hops away from $f(p)$ in $H''$, i.e., there is a path between $v$ and $f(p)$ in $H''$ containing at most $\ell$ edges. Here, $\gamma \ge 1$ is a large enough constant (it suffices to set $\gamma = 1600$). 

\subparagraph*{Long-path case.}

First, we define a set of {important} points, referred to as the landmarks, which will help us divide the points of $V$ into equivalence classes. In particular, each equivalence class would contain a subset of points of $V$ whose distance vectors w.r.t. a fixed sequence of landmarks are approximately equal. Moreover, the centroid set we will construct contains exactly one representative point from each such equivalence class, essentially leading to a small-sized set. To this end, we will then define a notion of \emph{rounded distances} to landmarks, and use it to construct $\bbCL$ such that for each possible distance tuple, $\bbCL$ contains a point having that \emph{rounded} distance tuple.    

\medskip\noindent\textbf{Landmarks.} Let $0 \le \mu \le \frac{1}{2}$ be small enough -- later we will define $\mu$ to be a function of $\epsilon$ and $z$. Throughout the following discussion, fix a root-leaf path $R_1, R_2, \ldots, R_t$ in $\T$, where $R_1$ is the root region, and $R_t$ is a leaf-region. Without loss of generality, we assume that $R_t \cap X \neq \emptyset$ -- otherwise we prune such a root-leaf path at the lowest region that contains at least one vertex of $X$. Note that $t = \Oh(\log |X|)$. We recall some notation from \Cref{subsec:decomp}. Let $G_i = G[R_i]$ be an induced subgraph corresponding to region $R_i$, and let $H_i$ be a constant-stretch spanner of $G_i$. For $1 \le i \le t-1$, let $\P_i$ be the set of shortest-path separators obtained by applying \Cref{thm:sp-separator} to obtain the children of $R_i$ in the recursive decomposition, where $\P_i = \{P^{i, 1}, P^{i, 2}, \ldots, P^{i, b}\}$, with $b = \Oh(1)$. Let $\P = \bigcup_{i = 1}^{t-1} \P_i$. Finally, we define some notational shorthand: let $d_i(\cdot, \cdot) = d_{G_i}(\cdot, \cdot)$, and $\d(\cdot, \cdot) = d_G(\cdot, \cdot)$. For a path $P\ij$, and a point $q_1 \in X$, we consider two ways of defining rounded distance tuples, that will be useful in different cases.

\textbf{First rounded distance tuple.} Consider a non-leaf region $R_i$ and a shortest path $P\ij \in \P_i$. For every pair of vertices $q_1, q_2 \in X$, let $D(i, q_1, q_2) \coloneqq \d_i(q_1, q_2) + \d(q_2, \A)$. We fix $i$ and points $q_1, q_2$ until the end of the current discussion about first rounded distance tuple, and use $D \coloneqq D(i, q_1, q_2)$ for brevity. 

Let $Q\ij(q_1, q_2)$ denote the set of vertices $v$ on the path $P\ij$ such that $\d_i(v, q_1) \le D/\mu^2$. Arbitrarily orienting $P\ij$ from one endpoint to another, let $u$ and $v$ denote the first and last vertex belonging to $Q\ij$ respectively. Note that, 
\begin{align*}
	\d_{P\ij}(u, v) &= \d_{H_i}(u, v) \tag{Since $P\ij$ is a shortest path in $H_i$}
	\\&\le \d_{H_i}(u, q_1) + \d_{H_i}(q_1, v) 
	\\&\le \alpha \cdot \lr{ \d_{i}(u, q_1) + \d_{i}(q_1, v) } 
	\\&\le \frac{2\alpha D}{\mu^2}.
\end{align*}
Now, we construct an $(\mu^2 \cdot D)$-net $\L(i, j, q_1, q_2)$ of $Q\ij(q_1, q_2)$. Formally, we treat the subpath between $u$ and $v$ as a real interval $\left[ 0, \beta \right]$, where $\beta \le \frac{2\alpha D}{\mu^2}$. We place a ``mark'' at every $\mu^2 D$ units of distance. Now we go back to the subpath, which is a discrete sequence of edges. If a mark coincides with a vertex along the path, then we add the vertex to the set $\L(i, j, q\ij)$. Otherwise, if a mark falls between two consecutive vertices along the path, then we add both of the vertices to the set $\L(i, j, q_1, q_2)$.

\begin{observation} \label{obs:landmarks}
	\begin{itemize}
		\item $|\L(i, j, q_1, q_2)| = \Oh \lr{\mu^{-4}}$ (since $\alpha$ is an absolute constant), and
		\item For any $x \in V(P\ij)$ with $\d_{G_i}(q_1, x) \le D/\mu^2$, there exists an $l \in \L(i, j, q_1, q_2)$ such that $\d_{G_i}(x, l) \le \d_{H_i}(x, l) = \d_{P\ij}(x, l) \le \mu^2 D$.
	\end{itemize}
\end{observation}

Let $R'_t = \LR{s \in R_t : |\pi_{G}(x, s)| > \ell \text{ for all } x \in X  }$. That is, $R'_t$ is the subset of $R'_t$ of points $s$ such that the hop-distance of $s$ from \emph{every point in $X$} along the shortest path in $G$, is strictly larger than $\ell$. 

For a path $P\ij \in \P_i$, for each $q_1, q_2 \in X$, and any point $s \in R'_t \subseteq R_t$, we consider the following \emph{rounded distance tuple} $\tilde{\d}^1_{i, j}(s, (q_1, q_2))$, defined as the tuple formed by listing the following entries in a fixed order:
\begin{itemize}
	\item For each $l \in \L(i, j, q_1, q_2)$: $\displaystyle \tilde{d}^1_{i, j}(s, l) = \min \LR{\near*{\frac{\d_i(s, l)}{\mu^2 D}}, \frac{3D}{\mu^2}}$.
	\item $\displaystyle \tilde{d}^1_{i, j}(s, q_1) = \LR{ \near*{\frac{\d_i(s, q_1)}{\mu D}}, \frac{3D}{\mu} }$.
	\item For each $x \in X \cap R_t$: $\tilde{d}^1_{i, j}(s, x) = \LR{ \near*{\frac{\d_i(x, s)}{\mu \d(x, \A)}}, \frac{1}{\mu} \d(x, \A) }$.
\end{itemize}
Here, we use the notation $\near{\frac{p}{q}}$ to denote the integral multiple of $q$ that is closest to $\frac{p}{q}$.

For a leaf region $R_t$, there are no corresponding shortest path separators. Thus, we consider a different definition first rounded distance tuple. Let $\L(t) = R_t \cap X$. For each $p \in \L(t)$, let $\tilde{d}_t(s, x) = \LR{ \near*{\frac{\d(x, s)}{\mu \d(x, \A)}}, \frac{1}{\mu} \d(x, \A) }$. Let $\tilde{\d}^1_t(s, t)$ be this tuple.

\textbf{Second rounded distance tuple.} Again, fix a path $P\ij \in \P_i$, and a point $q_1 \in X$, and a pair $q_3, q_4 \in R_i \cup \{\bot\}$, where $\bot$ means undefined. Furthermore, we assume that $q_3$ and $q_4$ both are not undefined at the same time. Then, for a point $s \in R'_t \subseteq R_t$, we define the second rounded distance tuple $\tilde{d}^2_{i, j}(s, (q_1, q_3, q_4))$, defined as follows.

\begin{itemize}
	\item $\tilde{d}^2_{i, j}(s, (q_1, q_3, q_4)) = \btrue$ if $$\frac{1}{\mu} \cdot \lr{\d_i(q_1, q_4) + \d(q_4, \A)} < \d_i(q_1, s) < \mu \cdot \lr{\d_i(q_1, q_3) + \d(q_3, \A)}$$
	Here, if $q_4 = \bot$, then the first inequality is omitted from the definition, and if $q_3 = \bot$, then the second inequality is omitted from the definition. Note that due to our assumption, at least one of $q_3$ and $q_4$ is not undefined, so the definition cannot be vacuously true.
	\item $\tilde{d}^2_{i, j}(s, (q_1, q_3, q_4)) = \bfalse$ otherwise, i.e., when $q_4 \neq \bot$ and the first inequality does not hold, or when $q_3 \neq \bot$ and the second inequality does not hold. 
\end{itemize}

\textbf{Construction of $\bbCL$. } For each root-leaf path $R_1, \ldots, R_t$, for each $P\ij \in \P$, and we of the two choices of rounding of distances (i.e., $\tilde{d}^1_{i, j}(\cdot, \cdot)$ or $\tilde{d}^2_{i, j}(\cdot, \cdot)$), we consider all possible values that the entries in the corresponding rounded tuples may take. We also consider the rounded distance tuple $\tilde{\d}^1_t(\cdot)$. For each possible choice of rounding type and of the entries, we obtain a large rounded distance tuple $\tilde{\mathsf{D}}$ by concatenating all tuples thus constructed in a fixed order. That is, $\tilde{\mathsf{D}}$ is obtained by corresponding rounded distance tuples (i.e., either the potential numerical or \btrue/\bfalse\ values that the entries in $\tilde{d}^1_{\cdot}(\cdot)$ or $\tilde{d}^2_{\cdot}(\cdot)$ may take respectively). If there exists a point $s \in R'_t \subseteq R_t$ that achieves all the corresponding entries (i.e., rounded distance values and/or \btrue/\bfalse), as specified by $\tilde{\mathsf{D}}$, then we add it to $\bbCL$.


\newcommand{\qi}{q^i_1}
\newcommand{\qii}{q^i_2}
\newcommand{\qit}{q^i_3}
\newcommand{\qiv}{q^i_4}



\subsection{Construction of the Approximate Solution} \label{subsec:approx-construction}
We now show that for any solution $\mathcal{S}$, every center in $\calS$ can be approximated by a centroid from $\mathbb{C}$. Consider any center $s \in \mathcal{S}$, and let $X_s \subseteq X$ be the subset of points whose closest center in $S$ is $s$ (breaking ties arbitrarily). First, we show how to find $\tils \in \mathbb{C}$ such that for any point $p \in X_s$ with $\cost(p, s) \le \lr{\frac{10z}{\epsilon}}^z \cdot \cost(p, \A)$, $ \cost(p, \tils)\le \cost(p, s)  + 3\epsilon \lr{ \cost(p, s) + \cost(p, \A) }$.

Now, we show how to find a replacement center for any $s \in \mathcal{S}$, which we denote by $\rho(s)$. We consider different cases to find $\rho(s)$. The following cases, or \emph{replacement rules}, are applied in the following manner. We proceed to each of the replacement rules in the order in which they are presented. For a particular replacement rule, we iterate over all centers $s \in \calS$ for which replacement has not yet been found using previous rules, and check whether the current rule is applicable to $s$, in which case we define $\rho(s)$ as specified by the rule. Thus, at the end of the iteration, we are left with the centers $s \in \calS$ for which none of the previous rules are applicable. Then, we proceed to the next rule and proceed in a similar manner.

\textbf{Replacement using $\bbCN$.} 

\textbf{Case 1.} Suppose there exists a point $p \in X_s$ with $\d_G(p, \A) < 1$.
\\Let $X'_s \coloneqq \{q \in X_s : \d_G(q, \A) < 1 \}$.
Note that $p \in X'_s$, which implies that $X'_s \neq \emptyset$. Let $p_i \coloneqq \argmin_{\{q\in X'_s\}} \d_G(q,\A)+\d_G(q,s)$. Note that $s \in B_i = B(p_i, (10z/\epsilon) \cdot \d_G(p_i, \A) )$. Let $\tils$ be the closest point to $s$ in the $(\epsilon^3/z^3)\cdot \d_G(p_i,\A)$-net $B'_i$ constructed for $B_i$. Note that such an $\tils$ was added to $\bbCN$. In this case, we add $\tils \in \bbCN \subseteq \bbC$ to $\tilde{\mathcal{S}}$, and set $\rho(s) = \tils \in \bbCN$.  Define $\Snet$ as the subset of $\calS$ as all the centers $s$ whose replacement was found using this manner. 

\textbf{Case 2.} Consider the following case.
\begin{enumerate}
	\item Suppose for all points $p \in X_s$, it holds that $\d_G(p, \A) \ge 1$, and thus, we have not found $\rho(s)$ in the previous case,
	\item $f(s) \in V(H'')$ was added to $\bbCS \subseteq \bbC$, and
	\item There exists a point $q \in X_{s_q}$ where $s \neq s_q \in \Snet$, and $\tilde{s}_q = \rho(s_q) \in \bbCN$ with the following properties: (1) $\d_G(q, s_q) \le \epsilon/z$, (2) $\d_G(q, \tilde{s}_q) > \d_G(q, f(s))$.  
\end{enumerate}
In this case, we let $\rho(s) = \tilde{s}_q$, as defined above (if there are multiple choices for $q$, and thus for $s_q$, we may choose arbitrarily). We define $\Snetsub$ as the subset of $\calS$ as all the centers $s$ whose replacement was found using this case. Note that for centers $s$ in $\Snet$ as well as $\Snetsub$, $\rho(s)$ belongs to $\bbCN$; but the difference is that, for a center $s \in \Snetsub$, we make use of a center in $\tilS$ that was already found as a replacement for a \emph{different} center $s_q \in \Snet$.

\textbf{Replacement using $\bbCS$.} Consider the case when the first two items in the previous case hold, but there is no point $p$ satisfying the conditions in the third item. That is, suppose $s$ satisfies the following conditions:
\begin{enumerate}
	\item Suppose for all points $q \in X_s$, it holds that $\d_G(q, \A) \ge 1$, 
	\item $f(s) \in V(H'')$ was added to $\bbCS \subseteq \bbC$, and
	\item There exists no point $p \in X_{s_p}$ where $s \neq s_p \in \calS$, and $s'_p = \rho(s_p) \in \bbCN$ with the following properties: (1) $\d_G(p, s_p) \le \epsilon/z$, and (2) $\d_G(p, s'_p) > \d_G(p, f(s))$. 
\end{enumerate}
Then, we add $f(s)$ to $\tilS$, and set $\rho(s) = f(s)$. Let $\Ssup \subseteq \calS$ denote the set of centers $s$ such that $\rho(s) \in \bbCS$. 

We let $\Sland \coloneqq \calS \setminus (\Snet \cup \Snetsub \cup \Ssup)$ denote the subset of centers whose replacement has not been found in either $\bbCN$ or in $\bbCS$ as specified in the previous paragraph. Here, the notation anticipates that for such centers $s \in \Sland$, we will find $\rho(s)$  using $\bbCL$. Before that, we first note the following observation.
\begin{observation} \label{obs:long-path}
	For any $s \in \Sland$ and for any point $p \in X$, $|\pi_G(p, s)| > \ell$.
\end{observation}
\begin{proof}
	Suppose for contradiction that for some $s \in \Sland$ and $p \in X$ $p,u_1,\ldots,u_t,s$ is a path $\pi_G(p,s)$. Then the path $f(p),f(u_1),\ldots,f(u_t),f(s)$ is in $H''$. Thus, $f(s)$ is at most $\ell$ hops away from $f(p)$, and hence $f(s)$ should have been added to $\bbC$ by our construction. But, this is a contradiction, and hence $|\pi_G(p,s)| > \ell$. 
\end{proof}

\textbf{Replacement using $\bbCL$.} Now, we show how to find a replacement for $s \in \Sland$ using $\bbCL$. For this, we make use of the recursive decomposition tree $\T$ and the set $\bbCL$ constructed previously. Fix the center $s \in \Sland$ for the following discussion, and let $R_t$ be the leaf node of $\T$ containing $s$. Let $\Pi = (R_1, R_2, \ldots, R_t)$. Note that \Cref{obs:long-path} implies that $s \in R'_t \subseteq R_t$.  

Consider a non-leaf region $R_i$, and a path $P\ij \in \P_i$. Again, we use the shorthand $\d_{i} = \d_{G_i}$, and $\d = \d_G$. Recall that $G_i = G[R_i]$. Let $q^i_1 = \arg\min_{p \in R_i} \d_i(p, s)$. Now, we consider two cases.
\begin{enumerate}
	\item There exists some $q^i_2 \in R_i$ such that $\mu \d_i(q^i_1, s) \le \d_i(q^i_1, q^i_2) + \d(\qii, \A) \le \frac{1}{\mu} \d_i(\qii, s)$.
	\\In this case, we pick the rounding $\tilde{\d}^1_{i, j}(s, (\qi, \qii))$. 
	\item Otherwise, we proceed as follows.
	\\(i) If there exists a point $q \in R_i$ such that (i) $\d_i(\qi, q) + \d(q, \A) > \frac{1}{\mu} \d_i(\qi, s)$. Then, let $\qit = \arg\min_{q \in R_i} \d_i(\qi, q) + \d(q, \A)$. If there exists no such $q \in R_i$, then let $\qit = \bot$, i.e., undefined.
	\\(ii) If there exists a point $q \in R_i$ such that (i) $\d_i(\qi, q) + \d(q, \A) < \mu \d_i(\qi, s)$. Then, let $\qiv = \arg\max_{q \in R_i} \d_i(\qi, q) + \d(q, \A)$. If there exists no such $q \in R_i$, then let $\qiv = \bot$, i.e., undefined.
	\\Note that at least one of $\qit$ and $\qiv$ is not undefined. In this case, we pick the rounding $\tilde{\d}^2_{i, j}(s, (\qi, \qit, \qiv))$. 
\end{enumerate}
We obtain a rounded distance tuple $\tilde{\mathsf{D}}(s)$ by concatenating over all paths $P\ij \in \P$, the rounded distance tuples $\tilde{d}^1_{i, j}(s, \cdot)$, or $\tilde{d}^2_{i, j}(s, \cdot)$ as defined above. Since $s \in R'_t$ and has the rounded distance tuple $\tilde{\mathsf{D}}(s)$, we conclude that $\bbCL \cap R'_t$ must contain a point, say $\tils$, that also has the rounded distance tuple $\tilde{\mathsf{D}}(s)$. We let $\rho(s) = \tils$. 

This completes the construction of the set $\tilS \subseteq \bbC$. In the next subsection, we analyze the properties of this solution.

\subsection{Error Analysis} \label{subsec:approx-error}

In this section, we show how to bound the error for any relevant point $p\in X_s$. The overarching goal is to show that, for any point $p \in X$ such that $\d_G(p, \calS) \le \frac{\gamma' z}{\epsilon} \cdot \d_G(p, \A)$, or $\d_G(p, \tilS) \le \frac{\gamma' z}{\epsilon} \cdot \d_G(p, \A)$, it holds that $|\cost(p, \calS) - \cost(p, \tilS)| \le \Oh(\epsilon) \cdot (\cost(p, \calS) + \cost(p, \A))$. Since we the set $\tilS$ is constructed by careful examination of different cases, the proof of this claim is also based on exhaustive case analysis, which is organized into different lemmas. Before we proceed to formally state and prove these lemmas, we start with a high-level overview of the conceptual flow of the argument.

\textbf{Forward direction.} For a point $p \in X$ with $s$ being its closest center in $\calS$, we consider different cases based on whether $s$ belongs to $\Snet$, $\Ssup$, or $\Sland$. This corresponds to \Cref{cl:s-net}, \Cref{cl:s-net-sub1}, \Cref{cl:s-suppport-1}, and \Cref{lem:landmark-bound}, respectively. Here, we show that for some $\mu \ge 0$, either $\d_G(p, s) > \frac{\mu z}{\epsilon} \cdot \d_G(p, \A)$, or $\d_G(p, \rho(s)) \le (1+\epsilon/z) \cdot \d_G(p, s) + \Oh(\epsilon/z) \cdot \d_G(p, \A)$. This comprises of the forward direction of the proof.

\textbf{Reverse direction.} For a point $p \in X$ with $s\in \calS$ and $\tils \in \tilS$ being its closest centers in the two sets respectively. If $\tils = \rho(s)$, then it is relatively straightforward to argue that, if $\d_G(p, \tils) \le \frac{\mu z}{\epsilon} \cdot \d_G(p, \A)$, then for some $\mu \ge 0$, $\d_G(p, s) \le (1+\epsilon/z) \cdot \d_G(p, \tils) + \Oh(\epsilon/z) \cdot \d_G(p, \A)$. In another case, we may have that $\tils = \rho(s')$, where $s' \neq s$. That is, $\tils \in \tilS$ was found as a replacement for a center $s'$ that is \emph{not} the closest center in $\tilS$ to the point $p$. Here, the subtle possibility is that $\d_G(p, \tils) \ll \d_G(p, s)$, in which case we cannot hope to show the required bound. To argue that this does not happen, we carefully examine the different cases used to find $\rho(s)$ and $\rho(s')$, the replacements for $s$ and $s'$, respectively. The reverse direction comprises of \Cref{cl:backward-1}, \Cref{cl:backward-2} and \Cref{lem:landmark-bound}.

Finally, we combine the two directions, and show the required error bound on the solution in \Cref{lem:errorbound}.

\subsubsection*{Forward Direction}

\begin{claim} \label{cl:s-net}
	Consider a point $p \in X_s$ with $s \in \Snet$ with $\d_G(p, s) \le \frac{10z}{\epsilon} \d_G(p, \A)$. Then, for $\rho(s) = \tils \in \bbCN$, it holds that $|\d_G(p, \tils) - \d_G(p, s)|\le  (\epsilon/z) \cdot (\d_G(p, s)+\d_G(p, \A))$.
\end{claim}
\begin{proof}
	Since $s \in \Snet$, the set $X'_s = \{ q \in X_s : \d_G(q, \A) < 1\}$
	is non-empty. Then, we choose $p_i = \argmin_{q \in X'_s} \d_G(q, \A) + \d_G(q, s)$. Note that $s\in B_i=B(p_i,(10z/\epsilon)\cdot \d_G(p_i,\A))$, and a point $\tils$ closest to $s$ from the $(\epsilon^3/z^3)\cdot\d_G(p_i,\A)$-net of $B_i$ is added to $\bbCN$. Now, $\d_G(s, \tils) \le (\epsilon^3/z^3)\cdot\d_G(p_i,\A)< (\epsilon/z)$. 
	
	Now, there are two possibilities. If $\d_G(p, \A) \ge 1$, then $|\d_G(p, \tils) - \d_G(p, s)|\le \epsilon/z \le  (\epsilon/z)\cdot  \d_G(p, \A)$. 
	
	Otherwise, if $\d_G(p, \A) < 1$, then $p \in X'_s$. Therefore, $\d_G(p_i, s) + \d_G(p_i, \A) \le \d_G(p, s) + \d_G(p, \A)$. Hence, $|\d_G(p, \tils) - \d_G(p, s)|\le \frac{\epsilon}{z} \cdot \d_G(p_i, \A) \le \frac{\epsilon}{z} (\d_G(p, s) + \d_G(p, \A))$.
\end{proof}

\begin{claim} \label{cl:s-net-sub1}
	Consider a point $p \in X_s$ for some $s \in \Snetsub$, such that $\d_G(p, s) \le \frac{10z}{\epsilon} \cdot \d_G(p, \A)$. Then, $\d_G(p, \tilS) \le \d_G(p, s) + \Oh(\epsilon/z) \cdot \d_G(p, \A)$. 
\end{claim}
\begin{proof}
	First, we claim that $\d_G(p, \A) \ge 1$. Suppose for the contradiction that $\d_G(p, \A) < 1$.
	Then, we would have found a replacement for $s$ in $\bbCN$ using Case 1, i.e., $s \in \Snet$. However, since $s \in \Snetsub$, Case 1 is not applicable, which is a contradiction.
	
	Since $s \in \Snetsub$, there exists a point $f(s) \in \bbCS$ such that $\d_G(s, f(s)) = \Oh(\epsilon^2/z^2)$. Furthermore, there exists a point $q \in X_{s_q}$ for some $s_q \in \Snet$ such that $\tilde{s}_q = \rho(s_q)$ satisfies the following properties: (1) $\d_G(q, s_q) \le \epsilon/z$, (2) $\d_G(q, \tilde{s}_q) > \d_G(q, f(s))$. Then, 
	\begin{align*}
		\d_G(p,\tilde{s}_q ) &\le \d_{G}(p, s)+\d_{G}(s, f(s))+\d_G(f(s),q) + \d_G(q, \tilde{s}_q)
		\\&\le \d_{G}(p, s)+ \Oh(\epsilon^2/z^2)+2 \cdot \d_G(q,\tilde{s}_q) \tag{$\because\ \d_G(q, f(s)) \le \d_G(q, \tilde{s}_q)$} 
		\\&\le \d_{G}(p, s)+ \Oh(\epsilon^2/z^2)+2 \cdot \lr{ \d_G(q, s_q) + \d_G(s_q,\tilde{s}_q) } 
		\\&\le \d_{G}(p, s)+\Oh(\epsilon^2/z^2) + 2 \cdot \epsilon/z \tag{$\because\ d(q,s_q)\le \epsilon/z$ and $\d_G(s_q,{s'}_q) \le \epsilon^3/z^3$}
		\\&\le \d_G(p, s) + \Oh(\epsilon/z) \cdot \d_G(p, \A) \tag{$\because\ \d_G(p,\A)\ge 1$.  }
	\end{align*}
\end{proof}

\begin{claim} \label{cl:s-suppport-1}
	Consider a point $p \in X_s$ for some $s \in \Ssup$ such that $\d_G(p, s) \le \frac{10z}{\epsilon} \cdot \d_G(p, \A)$. Then, $\d_G(p, \tilS) \le \d_G(p, s) + \Oh(\epsilon/z) \cdot \d_G(p, \A)$. 
\end{claim}
\begin{proof}
	Arguing as in the proof of \Cref{cl:s-net-sub1}, we first observe that $\d_G(p, \A) \ge 1$, otherwise $s$ would belong to $\Snet$, which is a contradiction. Since $s \in \Ssup$, $f(s)$ has been added to $\bbCS$, and we set $\rho(s) = f(s)$. Thus, $G$ has a path between $p$ and $f(s)$ of length exactly $|\pi_G(p,s)|+1$ whose weight is $\d_{G}(p, s)+\Oh(\epsilon^2/z^2)$, namely the path obtained by appending the edge $\{s,f(s)\}$ with the path $\pi_G(p,s)$. Hence, 
	\[\d_G(p,f(s))\le \d_{G}(p, s)+\Oh(\epsilon^2/z^2)\le \d_{G}(p, s)+\Oh(\epsilon^2/z^2)\cdot \d_G(p,\A)\]
	The last inequality follows, as $\d_G(p,\A) \ge 1$.
\end{proof}

Now, consider a point $p \in X_s$ with $\d_G(p, s) \le \frac{10z}{\epsilon} \d_G(p, \A)$. Now, if none of the \Cref{cl:s-net}, \Cref{cl:s-net-sub1}, and \Cref{cl:s-suppport-1} is applicable, then it follows that $s \not\in (\Snet \cup \Snetsub \cup \Ssup)$. That is, $s \in \Sland$. From \Cref{obs:long-path}, it follows that for all points $q \in X$, the hop-length of the path $\pi_G(q, s)$, i.e., $|\pi_G(q, s)|$, is strictly larger than $\ell$. 

In the following lemma, we show that the replacement $\rho(s)$ fond for $s$ from the set $\bbCL$ has approximately the same distance to $p$ as $s$. In fact, the following lemma is stronger in the following two aspects. First, it does not require that $s$ is the closest center to $p$ in $s$ -- we show that this inequality holds for \emph{any} $s'$ and its replacement $\tils$, as long as the respective distances to $p$ are bounded by $\Oh(\epsilon/z) \cdot \d_G(p, \A)$. Secondly, the lemma shows the inequalities in the both directions. Both of these properties will be useful subsequently.

\begin{lemma} \label{lem:landmark-bound}
	Consider a center $s \in \calS'$, and let $\tils = \rho(s) \in \tilS$ be its replacement found from $\bbCL$. Then, there exists a constant $\gamma$, such that for any point $p \in X$, the following holds:
	\begin{itemize}
		\item Either $\d_G(p, s) \ge \frac{\gamma z \d_G(p, \A)}{\epsilon}$, or $\d_G(p, \tils) \le (1+\frac{\epsilon}{z}) \d_G(p, s) + \frac{\epsilon}{z} \cdot \d_G(p, \A)$, and
		\item Either $\d_G(p, \tils) \ge \frac{\gamma z \d_G(p, \A)}{\epsilon}$, or $\d_G(p, s) \le (1+\frac{\epsilon}{z}) \d_G(p, \tils) + \frac{\epsilon}{z} \cdot \d_G(p, \A)$.
	\end{itemize}
\end{lemma}
\begin{proof}
	Let $s_1 \in \{s, \tils\}$, and $s_2$ be the other choice. Let $R_t$ is the leaf node of $\T$ containing $s$ and $\tils$, and note that each node along the root-leaf path $\Pi = R_1, \ldots, R_t$  in $\T$, contains both $s$ and $s'$. Observe that $s, \tils \in R'_t$. Recall that $\mu$ is the parameter defined in the construction of $\bbCL$. We will later set the value of $\mu$ to be $\frac{\gamma z}{\epsilon}$. We fix the point $p \in X$ for the rest of the discussion, and consider different cases.
	
	\textbf{Leaf case: $p \in R_t \cap X$.} Since $p \in \L(t) = R_t \cap X$, the rounding in the tuple $\tilde{d}^1_t(\cdot)$ implies that, either $\d_G(p, s) > \frac{1}{\mu} \cdot \d_G(p, \A)$, in which case $\d_G(p, \tils) > \frac{1}{\mu} \cdot \d_G(p, \A)$, since $s$ and $\tils$ both have same rounded distance tuples. Otherwise, $|\d_G(p, s) - \d_G(p, \tils)| \le \mu \d_G(p, \A)$. 
	
	\textbf{Internal node case: $p \not \in R_t$.} In this case, we consider a region $R_i$ along $\Pi$ such that some path $P\ij \in \P_i$ separates $p$ and $s_1$, and $x \in P\ij$ is a separating vertex that lies between $u$ and $v$ (as defined in \Cref{subsec:decomp}). For the rest of the proof, we use the shorthand $\d(\cdot, \cdot) \coloneqq \d_{G}(\cdot, \cdot)$, and $\d_i(\cdot, \cdot) = \d_{G_i}(\cdot, \cdot)$.	Let $\qi = \arg\min_{q \in R_i} \d_i(q, s)$. We note some of the consequences of these definitions and that of \Cref{cl:separating-paths} in the following inequalities.
	\begin{align}
		\d_i(p, s_1) &= \d(p, s_1) = \d(p, u) + \d(u, v) + \d(v, s_1) \label{eqn:ps1}
		\\\d_i(u, v) &\le \d_{H_i}(u, v) \le \alpha \cdot \d_i(u, v) = \alpha \cdot \d(u, v) \label{eqn:ps2}
		\\\d_{H_i}(u, x) + \d_{H_i}(x, v) &= \d_{H_i}(u, v) \le \alpha \cdot \d_i(u, v) \le \alpha \cdot \d(u, v) \label{eqn:ps3}
		\\\d_{H_i}(u, v) &\le \mu \cdot \d_G(p, s) \label{eqn:ps4}
	\end{align}
	
	In this case, we consider different cases based on the type of rounding tuple chosen for $P\ij$ while finding replacement for $s$ using rounded distance tuples.

	\textbf{Case 1.} Suppose there exists a $q \in R_i$ such that 
	\begin{equation}
		\mu \cdot \d_i(\qi, s) \le \d_i(\qi, q) +  \d(q, \A) \le \tfrac{1}{\mu} \cdot \d_i(\qi, s) \label{eqn:case1-rounding}
	\end{equation} 
	Note that the case assumption implies that, while defining the rounded distance tuple $\tilde{\mathsf{D}}(s)$, we chose the first rounded tuple $\tilde{d}^1_{i, j}(s, (\qi, \qii))$ w.r.t.\ the path $P\ij$, for some $\qii$ satisfying (\ref{eqn:case1-rounding}).	For the analysis of discussion of case 1 and its subcases, we use $D \coloneqq D(i, \qi, \qii) = \d_i(\qi, \qii) + \d(\qii, \A)$
	
	\textbf{Case 1A. $\d_i(x, \qi) \le \tfrac{D}{\mu^2}$.} Since $x \in P\ij$ and $\d_i(x, \qi) \le \tfrac{D}{\mu^2}$, this implies that $x \in Q\ij(\qi, \qii)$. Therefore, using \Cref{obs:landmarks}, there exists some $l \in \L(i, j, \qi, \qii)$ such that $\d_i(x, l) \le \mu^2 D$. Now consider,
	\begin{align}
		\d_i(s, l) &\le \d_i(l, x) + \d_i(x, \qi) + \d_i(\qi, s) \nonumber
		\\&\le \mu^2 D + \frac{D}{\mu^2} + \frac{D}{\mu} \tag{$\d_i(\qi, s) \le \frac{D^2}{\mu}$ by case assumption}
		\\\implies\ \d_i(s, l)&\le \frac{3D^2}{\mu} \label{eqn:s-l-range}
	\end{align}
	Therefore, $\tilde{d}^1_{i, j}(l, s) = \tilde{d}^1_{i, j}(l, \tils) = \LR{\near*{\frac{\d_i(s, l)}{\mu^2 D}}}$, which implies that $|\d_i(s_1, l) - \d_i(s_2, l)| \le \mu^2 D$. Now, we consider:
	\begin{align*}
		\d_i(p, s_2) &\le \d_i(p, l) + \d_i(l, s_2)
		\\&\le \d_i(p, l) + \d_i(l, s_1) + \mu^2 D \tag{From above}
		\\&\le \d_i(p, x) + \d_i(x, l) + \d_i(x, l) + \d_i(x, s_1) + \mu^2 D 
		\\&\le \d_i(p, x) + \d_i(x, s_1) + 3\mu^2 D \tag{Since $\d_i(x, l) \le \mu^2 D$}
		\\&\le \d_i(p, u) + \d_i(u, x) + \d_i(x, v) + \d_i(v, s_1) + 3\mu^2 D
		\\&\le \d(p, u) + \d(v, s_1) + \d_{H_i}(u, x) + \d_{H_i}(x, v) + 3\mu^2 D \tag{Using (\ref{eqn:ps1})}
		\\&\le \d(p, s_1) + \mu \cdot \d(p, s_1) + 3\mu^2 D \tag{Using (\ref{eqn:ps4})}
		\\&\le (1+\mu) \cdot \d(p, s_1) + 3\mu \d_i(p, s) \tag{Using case assumption: (\ref{eqn:case1-rounding})}
	\end{align*}
	Substituting $s_1 = s$ and $s_2 = \tils$, we obtain:
	\begin{align}
		\d(p, \tils) \le \d_i(p, \tils) &\le (1+\mu) \cdot \d(p, s) + 3\mu \cdot \d_i(p, s) \nonumber \tag{First inequality follows since $G_i$ is a subgraph of $G$}
		\\&\le (1+4\mu) \cdot \d(p, s) \label{eqn:tils-s-1}
	\end{align}
	and substituting $s_1 = s$ and $s_1 = \tils$, we obtain
	\begin{align}
		\d_i(p, s) &\le (1+\epsilon) \d(p, \tils) + 3\mu \d_i(p, s) \nonumber
		\\(1-3\mu) \cdot \d(p, s) \le (1-3\mu) \cdot \d_i(p, s) &\le (1+\mu) \d(p, \tils) \tag{First inequality follows since $G_i$ is a subgraph of $G$}
		\\\implies\ \d(p, s) &\le \frac{1+\mu}{1-3\mu} \cdot \d(p, \tils) \le (1+13\mu) \cdot \d(p, \tils) \label{eqn:s-tils-1}
	\end{align}
	
	\textbf{Case 1B.} $\d_i(x, \qi) > \frac{D}{\mu^2}$. 
	First, (\ref{eqn:case1-rounding}) implies that $\d_i(\qi, s) \le \frac{1}{\mu} D$. Therefore, by the definition of rounded distance tuple $\tilde{d}^1_{i, j}(s, (\qi, \qii))$, we have the following inequality
	\begin{equation}
		\d_i(s, \qi) - \d_i(\tils, \qi)| \le \mu D \label{eqn:s-tils-qi}
	\end{equation}
	Therefore, consider:
	\begin{align*}
		\d_i(s, \tils) \le \d_i(s, \qi) + \d_i(\qi, \tils) &\le 2\d_i(\qi, s) + \mu D 
		\\&\le \frac{2}{\mu} D + \mu D \le \frac{3}{\mu} D
		\\&\le 3\mu \cdot \d_i(x, \qi) \tag{From assumption of Case 1B}
		\\&\le 3\mu \cdot (\d_i(x, s_1) + \d_i(s_1, s_2) + 
		d_i(s_2, \qi)) 
	\end{align*}
	By rearranging the last inequality, we obtain that:
	\begin{align*}
		\d_i(s, \tils) &\le \frac{3\mu}{(1-3\mu)} \cdot (\d_i(x, s_1) + \d_i(s_2, \qi))
		\\&\le 9\mu \cdot (\d_i(x, s_1) + \d_i(s, \qi) + \mu D) \tag{Using (\ref{eqn:s-tils-qi})}
		\\&\le 9\mu \cdot (\d_i(x, v) + \d_i(v, s_1) + \d_i(p, s) + \d_i(p, s)) \tag{$\mu D \le \d_i(\qi, s)$ by (\ref{eqn:case1-rounding}) and $\d_i(\qi, s) \le \d_i(p, s)$ by the choice of $\qi$}
		\\&\le 9\mu \cdot \lr{\d_{H_i}(x, v) + \d_{i}(v, s_1)  + 2\d_{i}(p, s)}
		\\&\le 9\mu \cdot \lr{ \d_{H_i}(u, v) + \d_i(p, s_1) + 2\d_i(p, s) } 
	\end{align*}
	Where the last inequality follows from the fact that, $u, \ldots, x, \ldots, v$ is a shortest path in $H_i$, and $p, \ldots, v, \ldots, s_1$ is a shortest path in $G$ as well as $G_i$. Therefore,
	\begin{align*}
		\d_i(s, \tils) &\le 9\mu \cdot \lr{\mu \cdot \d(p, s_1) + 3\d_i(p, s)} \le 27\mu \cdot (\d(p, s_1) + \d_i(p,s)) 
	\end{align*}
	Therefore, 
	\begin{align}
		\d_i(p, s_2) &\le \d_i(p, s_1) + 27\mu \cdot (\d(p, s_1) + \d_i(p,s)) \label{eqn:s-tils-case1b}
	\end{align}
	By substituting $s_1 = s$, and $s_2 = \tils$ in (\ref{eqn:s-tils-case1b}), we obtain:
	\begin{align}
		\d(p, \tils) &\le \d_i(p, s) + 27\mu (\d_G(p, s_1) + \d_i(p,s)) \nonumber
		\\&= \d(p, s) + 54\mu \cdot \d(p, s) = (1+54\mu) \d(p, s)  \label{eqn:tils-s-2}
	\end{align}
	Where we use that $\d_i(p, s_1) = \d(p, s_1) = \d(p, s)$. Then, by substituting $s_1 = \tils$ and $s_2 = s$ in (\ref{eqn:s-tils-case1b}), we obtain:
	\begin{align}
		\d_i(p, s) &\le \d_i(p, \tils) + 27\mu \cdot (\d(p, \tils) + \d_i(p, s) ) \nonumber
		\\\implies (1-27\mu) \d(p, s) \le (1-27\mu) \d_i(p, s) &\le \d(p, \tils) + 27\mu \cdot \d(p, \tils) \tag{Since $\d_i(p, s_1) = \d(p, s_1) = \d(p, \tils)$}
		\\\implies \d(p, s) &\le \frac{1+27\mu}{1-27\mu} \cdot \d(p, \tils) \le (1+1539\mu) \d(p, \tils) \label{eqn:s-tils-2}
	\end{align}
	
	\textbf{Case 2.} Suppose there exists no point $q \in R_i$ satisfying (\ref{eqn:case1-rounding}). 
	Then, for any point $q$ (and in particular, for $q = p$), exactly one of the following inequalities holds:
	\begin{align}
		\d(q, \A) + \d_i(q, \qi) > \frac{1}{\mu} \d_i(\qi, s) \label{eqn:ql}
		\\\textbf{ OR }\d(q, \A) + \d_i(q, \qi) < \mu\d_i(\qi, s) \label{eqn:qr}
	\end{align}
	In this case, while finding replacement for $s$, we choose the second rounded distance tuple $\tilde{d}^2_{i, j}(s, (\qi, \qit, \qiv))$ in the rounded tuple $\tilde{\mathsf{D}}(s)$, with the following properties. If there exists a point $q \in R_i$ satisfying (\ref{eqn:ql}), then $\qit$ is chosen to be the one with smallest value on the left hand side; otherwise $\qit = \bot$. Analogously, if there exists a point $q \in R_i$ satisfying (\ref{eqn:qr}), then $\qiv$ is chosen to be the one with largest value on the left hand side; otherwise $\qiv = \bot$. Note that at least one of the two points $\qit$ and $\qiv$ is not equal to $\bot$ due to case assumption. Thus, the following inequality (inequalities) holds:
	\begin{align}
		\frac{1}{\mu} \cdot \lr{ \d_i(\qi, \qiv) + \d(\qiv, \A) } \le \d_i(\qi, s) \le \mu \cdot \lr{ \d_i(\qi, \qit) + \d(\qit, \A) } \label{eqn:s-q3-q4}
	\end{align}
	where, we drop an inequality from the requirement if the corresponding point $\qit$ or $\qiv$ is $\bot$. Thus, $\tilde{\d}^2_{i, j}(s, (\qi, \qit, \qiv)) = \btrue$. Therefore, $\tilde{\d}^2_{i, j}(\tils, (\qi, \qit, \qiv)) = \btrue$, which implies that: 
	\begin{align}
		\frac{1}{\mu} \cdot \lr{ \d_i(\qi, \qiv) + \d(\qiv, \A) } \le \d_i(\qi, \tils) \le \mu \cdot \lr{ \d_i(\qi, \qit) + \d(\qit, \A) } \label{eqn:tils-q3-q4}
	\end{align}
	Again, with the same caveat about dropping the appropriate inequality. 
	
	Now we consider two cases depending on which inequality from (\ref{eqn:ql}) and (\ref{eqn:qr}) is satisfied by $p$.
	
	\textbf{Case 2A.} Suppose $\d(p, \A) + \d_i(p, \qi) > \frac{1}{\mu} \d_i(\qi, s)$. Then, by the choice of $\qit$, we have that:
	\begin{align*}
		\d_i(s, \tils) &\le \d_i(s, \qi) + \d_i(\tils, \qi)
		\\&\le 2\mu \cdot \lr{ \d_i(\qi, \qit) + \d(\qit, \A) } \tag{From (\ref{eqn:s-q3-q4}) and (\ref{eqn:tils-q3-q4})}
		\\&\le 2\mu \cdot \lr{ \d_i(\qi, p) + \d(p, \A) } \tag{Since $\qit$ is a point minimizing the LHS of (\ref{eqn:ql})}
		\\&\le 2\mu \cdot \lr{ \d_i(p, s) + \d_i(q_i, s) + \d(p, \A) } 
		\\&\le 2\mu \cdot \lr{ 2\d_i(p, s) + \d(p, \A) } \tag{Since $\qi$ is a point minimizing $\d_i(\cdot, s)$}
		\\&\le 4\mu (\d_i(p, s) + \d(p, \A))
	\end{align*}
	Therefore, we obtain the following inequality.
	\begin{align}
		\d(p, s_1) \le \d(p, s_2) + \d(s_1, s_2) &\le \d(p, s_2) + \d_i(s_1, s_2) \nonumber
		\\&\d(p, s_2) + 4\mu \cdot (\d_i(p, s) + \d(p, \A)) \label{eqn:s1s2-case2A}
	\end{align}
	By plugging in $s_1 = s$ and $s_2 = \tils$, we obtain:
	\begin{align}
		\d(p, s) &\le \d(p, \tils) + 4\mu \cdot \d(p, s) + 4\mu \cdot \d(p, \A) \nonumber
		\\\implies\ (1-4\mu) \cdot \d(p, s) &\le \d(p, \tils) + 4\mu \cdot \d(p, \A) \nonumber
		\\\implies\ \d(p, s) &\le (1+8\mu) \cdot \d(p, \tils) + 12\mu \cdot \d(p, \A) \label{eqn:s-tils-3}
	\end{align}
	And by plugging in $s_1  = \tils$ and $s_2 = s$, we obtain:
	\begin{align}
		\d(p, \tils) &\le (1+4\mu) \cdot \d(p, s) + 4\mu \cdot \d(p, \A) \label{eqn:tils-s-3}
	\end{align}
	
	\textbf{Case 2B.} Suppose $\d(p, \A) + \d_i(p, \qi) < \mu \d_i(\qi, s)$. Then, by the choice of $\qiv$, it holds that $\d_i(\qi, s_1) > \frac{1}{\mu} \cdot \lr{ \d_i(\qi, \qiv) + \d(\qiv, \A) }$, since both $s$ and $\tils$ satisfy this inequality. Therefore,
	\begin{align*}
		\d(p, s_1) = \d_i(p, s_1) &\ge \d_i(\qi, s_1) - \d_i(p, \qi) 
		\\&\ge \frac{1}{\mu} \cdot \lr{ \d_i(\qi, \qiv) + \d(\qiv, \A) } - \d_i(p, \qi) \tag{From above}
		\\&\ge \frac{1}{\mu} \cdot \lr{ \d_i(\qi, p) + \d(p, \A) } - \d_i(p, \qi) \tag{By the choice of $\qiv$}
		\\&\ge \frac{1}{\mu} \cdot \d(p, \A) 
	\end{align*}
	Thus, the above inequality shows that in case 2B, both $\d(p, s)$ and $\d(p, \tils)$ are larger than $\frac{1}{\mu} \d(p, \A)$. 
	
	Thus, by combining all the cases, we obtain that, either $\d(p, s_1) > \frac{1}{\mu} \d(p, \A)$, or $\d(p, s_2) \le (1+c\mu) \cdot \d(p, s_1) + c'\mu \d(p, s_2)$, for some constants $1 \le c, c' \le 1539$. Let $\gamma = \max\LR{c, c'} = 1539$. Therefore, by setting $\mu = \gamma \cdot \frac{\epsilon}{z}$, we obtain that: either $\d(p, s_1) > \frac{\gamma z \d(p, \A)}{\epsilon}$, or $\d(p, s_2) \le (1 + \frac{\epsilon}{z}) \d(p, s_1) + \frac{\epsilon}{z} \d(p, \A)$. 
\end{proof}

\subsubsection*{Reverse direction.}

Now, we aim to prove the ``reverse direction''. Specifically, we focus on a point $p \in X$ with $\tils_1 \in \tilS$ being its closest center in $\tilS$. Then we show that if $\d_G(p, \tils_1) \le \frac{10z}{\epsilon} \cdot \d_G(p, \A)$, then $\d_G(p, \calS) \le (1+\Oh(\epsilon/z)) \cdot \d_G(p, \tils_1) + \Oh(\epsilon/z) \cdot \d_G(p, \A)$. As before, this proof is organized in the following claims, based on different cases. Before this, we need the following lemma whose proof follows trivially from Lemma 18 of \cite{Cohen-AddadLSS22}.

\begin{lemma}\label{lem:no-problematic-pt}
	Consider any $s\in \Snet$ such that $X'_s = \{ q \in X_s : \d_G(q, \A) < 1
	\}$. Let $q' = \argmin_{\{ q \in X'_s \}} \d_G(q,\A)+\d_G(q,s)$. Then, without loss of generality, we can assume that for all points $p$ with $\d_G(p,\A) < 1$, $(\d_G(p,\A)+\d_G(p,s)) > \frac{\epsilon^2}{8z^2}\cdot (\d_G(q',\A)+\d_G(q',s))$.  
\end{lemma}

Thus, we assume the property guaranteed by \Cref{lem:no-problematic-pt}, and proceed to proving the following claims.

\begin{claim} \label{cl:backward-1}
	Consider a point $p \in X$ with (1) $\d_G(p, \A) < 1$, and (2) $\tils_1$ is the closest center in $\tilS$ to $p$, with $\d_G(p, \tils_1) \le \frac{10z}{\epsilon} \cdot \d_G(p, \A)$. Then, $\d_G(p, \calS) \le (1+\Oh(\epsilon/z)) \cdot \d_G(p, \tils_1) + \Oh(\epsilon/z) \cdot \d_G(p, \A)$.
\end{claim}
\begin{proof}
	Let $s$ be the closest center in $\calS$ to and from $p$. Note that since $\d_G(p, \A) < 1$, $s \in \Snet$, and $\rho(s) = \tils \in \bbCN$. Furthermore, by the condition of lemma, $\d_G(p, \tils_1) \le \d_G(p, \tils)$.
	We consider three cases.

	First, suppose there exists some $s_1 \in \Snet$ such that $\rho(s_1) = \tils_1$. \footnote{Note here that it may be the case that there also exists some $s' \in \Snetsub$ for which $\rho(s') = \tils_1$. But, in this analysis, the existence of $s_1 \in \Snet$ is sufficient.} Then, the set $X'_{s_1} = \{ q \in X_{s_1} : \d_G(q, \A) < 1
	\}$ is non-empty. Let $q' = \argmin_{q \in X'_{s_1}} \d_G(q, \A) + \d_G(q, s_1)$. It follows that $\d_G(s_1, \tils_1) \le \frac{\epsilon^3}{z^3} \cdot \d_G(q', \A)$. Then, by \Cref{lem:no-problematic-pt},
	\begin{align*}
		\d_G(p, s_1) &\le \d_G(p, \tils_1) + \d_G(s_1, \tils_1) 
		\\&\le \d_G(p, \tils_1) + \frac{\epsilon^3}{z^3} \cdot \d_G(q', \A)
		\\&\le \d_G(p, \tils_1) + \frac{\epsilon^3}{z^3} \cdot (\d_G(q', \A)+\d_G(q', s_1))
		\\&\le \d_G(p, \tils_1) + \frac{\epsilon^3}{z^3} \cdot \frac{8z^2}{\epsilon^2}\cdot (\d_G(p, \A)+\d_G(p, s_1))
		\\&= \d_G(p, \tils_1) + \Oh(\epsilon/z) \cdot (\d_G(p, \A)+\d_G(p, s_1))
	\end{align*}
	Hence, \[ \d_G(p, s) \le \d_G(p, s_1) \le (1+\Oh(\epsilon/z))\cdot \d_G(p, \tils_1)+\Oh(\epsilon/z)\cdot \d_G(p, \A).\]

	In the second case, there exists some $s_1 \in \Ssup$ such that $\rho(s_1) = f(s_1) = \tils_1$. 
	\[\d_G(p,\tils_1)\ge \d_G(p,s_1)-\Oh(\epsilon^2/z^2)\ge \d_G(p,s)-\Oh(\epsilon^2/z^2)\]
	Now, as $s_1$ was replaced by $\tils_1\in\Ssup$ although $\d_G(p,\tils)> \d_G(p,\tils_1)$, it must be that $\d_G(p,s)> \epsilon/z$, by our replacement scheme. Thus
	\begin{align*}
		\d_G(p, s) \le \d_G(p, s_1) &\le \d_G(p, \tils_1) + \d_G(s_1, \tils_1)
		\\&\le \d_G(p, \tils_1) + \Oh(\epsilon^2/z^2)
		\\&\le \d_G(p, \tils_1) + \Oh(\epsilon/z) \cdot \d_G(p, s)
	\end{align*}
	Therefore, $\d_G(p, s) \le (1+\Oh(\epsilon/z)) \cdot \d_G(p, \tils_1) \le (1+\Oh(\epsilon/z)) \cdot \d_G(p, \tils_1) + (\epsilon/z) \cdot \d_G(p, \A)$.
\end{proof}

\begin{claim} \label{cl:backward-2}
	Consider a point $p \in X$ with $\d_G(p, \A) \ge 1$, and $\tils_1$ is the closest center to $p$ from $\tilS$ such that $\rho(s_1) = \tils_1$, where $s_1 \in \Snet \cup \Ssup$. Then, $\d_G(p, \calS) \le \d_G(p, \tils_1) + \frac{\epsilon}{z} \cdot \d_G(p, \A)$. 
\end{claim}
\begin{proof}
	We consider two cases. If $s_1 \in \Snet$, then $\d_G(s_1, \tils_1) \le \frac{\epsilon^3}{z^3} \le \frac{\epsilon}{z} \cdot \d_G(p, \A)$. Alternatively, if $s_1 \in \Ssup$, then $\d_G(s_1, \tils_1) \le \frac{\epsilon^2}{z^2} \le \frac{\epsilon^2}{z^2} \cdot \d_G(p, \A)$. Here, we use the fact that $\d_G(p, \A) \ge 1$ in each of the inequalities. Then, the claim follows since:
	$$ \d_G(p, \calS) \le \d_G(p, s_1) \le \d_G(p, \tils_1) + \d_G(s_1, \tils_1) \le \d_G(p, \tils_1) + \frac{\epsilon}{z} \cdot \d_G(p, \A).$$
\end{proof}

\paragraph{Putting Everything Together.} First, we prove the error bound in the following lemma.
\begin{lemma} \label{lem:errorbound}
	For any $\calS \in \bbC^k$, there exists a $\tilS \in \bbC^k$ with the following property. For any point $p \in X$, if $\cost(p, \calS) \le \lr{\frac{10z}{\epsilon}}^z \d(p, \A)$, or $\cost(p, \tilS) \le \lr{\frac{10z}{\epsilon}}^z \d(p, \A)$, then
	$$|\cost(p, \calS) - \cost(p, \tilS)| \le \frac{\epsilon}{z \log(z/\epsilon)} \cdot \lr{ \cost(p, \calS) + \cost(p, \A) }.$$
\end{lemma}
\begin{proof}
	Let $\mu \ge 0$ be a large enough constant.
	By combining \Cref{cl:s-net}, \Cref{cl:s-net-sub1}, \Cref{cl:s-suppport-1}, and \Cref{lem:landmark-bound}, and by appropriately rescaling $\epsilon$, we conclude that for every point $p \in X$ with $\d_G(p, \calS) \le \frac{9\mu z}{\epsilon} \cdot \d_G(p, \A)$, the following holds.
	\begin{align}
		 \d_G(p, \tilS) &\le \lr{1 + \epsilon/z} \cdot \d_G(p, \calS) + (\epsilon/z) \cdot \d_G(p, \A) \label{ineq:dist-fwd}
		 \\\implies\ \cost(p, \tilS) &\le (1+\epsilon) \cdot \cost(p, \calS) + \epsilon \cdot \cost(p, \A) \label{ineq:cost-fwd}
	\end{align} 	
	Here, we use \Cref{lem:triangle} in the last step.

	Similarly, by combining \Cref{cl:backward-1}, \Cref{cl:backward-2}, and \Cref{lem:landmark-bound}, and appropriately rescaling $\epsilon$, we conclude that for every point $p \in X$ with $\d_G(p, \tilS) \le \frac{9\mu z}{\epsilon} \d_G(p, \A)$, the following holds.
	\begin{align}
		\d_G(p, \calS) &\le \lr{1 + \epsilon/z} \cdot \d_G(p, \tilS) + (\epsilon/z) \cdot \d_G(p, \A) \label{ineq:dist-bwd}
		\\\implies\ \cost(p, \tilS) &\le (1+\epsilon) \cdot \cost(p, \calS) + \epsilon \cdot \cost(p, \A) \label{ineq:cost-bwd}
	\end{align}
	Where we again use \Cref{lem:triangle} in the last step.
	
	Now, using (\ref{ineq:dist-fwd}), we can infer that for any $\lambda \le 9\mu$,
	\begin{equation}
		\d_G(p, \calS) \le \frac{\lambda z}{\epsilon} \cdot \d_G(p, \A)\ \ \implies\ \ \d_G(p, \tilS) \le \frac{3\lambda z}{\epsilon} \cdot \d_G(p, \A) \label{ineq:dist-fwd-1}
	\end{equation}
	Similarly, using (\ref{ineq:dist-bwd}), we can infer that for any $\lambda \le 9\mu$,
	\begin{equation}
		\d_G(p, \tilS) \le \frac{\lambda z}{\epsilon} \cdot \d_G(p, \A)\ \ \implies\ \ \d_G(p, \calS) \le \frac{3\lambda z}{\epsilon} \d_G(p, \A) \label{ineq:dist-bwd-1}
	\end{equation}
	Therefore, by applying (\ref{ineq:dist-fwd-1}) and (\ref{ineq:dist-bwd-1}) with $\lambda \gets \mu$, we can conclude that either both $\d_G(p, \calS)$ and $\d_G(p, \calS)$ are larger than $\frac{3\mu z}{\epsilon} \cdot \d_G(p, \A)$, and we are done. Otherwise, both are smaller than $\frac{9\mu z}{\epsilon} \cdot \d_G(p, \A)$. Then, by using (\ref{ineq:cost-fwd}) and (\ref{ineq:cost-bwd}), we obtain that for any such point $p \in X$,
	\begin{align}
		| \cost(p, \calS) - \cost(p, \tilS) | &\le 2\epsilon \cdot \lr{ \cost(p, \calS) + \cost(p, \tilS) + \cost(p, \A) } \nonumber
		\\\implies (1-2\epsilon) \cdot | \cost(p, \calS) - \cost(p, \tilS) | &\le 4\epsilon \cdot \lr{ \cost(p, \calS) + \cost(p, \A) } \tag{Using $\cost(p, \tilS) \le \cost(p, \calS) + |\cost(p, \tilS) - \cost(p, \calS)|$}
		\\\implies\ | \cost(p, \calS) - \cost(p, \tilS) | &\le 12\epsilon \lr{\cost(p, \calS) + \cost(p, \A)}
	\end{align}
	Thus, we obtain the lemma by a final rescaling $\epsilon \gets \frac{\epsilon}{c z \log(z/\epsilon)}$ for some large enough constant $c$.
\end{proof}
Now, we bound the size of the centroid set.
\begin{restatable}{lemma}{centroidsize}
	\label{lem:centroid-size}
	$|\mathbb{C}| = \exp 
	\lr{\Oh \lr{ \log^2 |X| + z^{16} \epsilon^{-8} (\log(z/\epsilon))^{8} \log|X| }}$. 
\end{restatable}
\begin{proof}
	Recall that $\bbC \coloneqq \bbCL \cup \bbCS \cup \bbCN$.
	First, the size of $\bbCS$ is bounded by $(\Delta(z^2/\epsilon^2))^{\ell}|X|$, due to \Cref{prop:support-graph}. Since $\ell = \Theta(z/\epsilon)$, this is equal to $(z/\epsilon)^{\Oh(z/\epsilon)} \cdot |X|$. Next, we bound $\bbCN$. The size of each $B'_i$, which is an $(\epsilon^3/z^3)\d_G(p_i,\A)$-net of $B_i$, can be upper bounded by $\Oh(z^8\epsilon^{-8})$ by using \Cref{prop:mu-net}. Hence, the size of $\bbCN$ is $\Oh(|X|z^8\epsilon^{-8})$. 
	
	Now, we bound $\bbCL$. Note that $\bbCL$ contains at most one point per rounded distance tuple and each leaf. Note that there are at most $\Oh(|X|)$ leaves, and thus the same number of root-leaf paths. Now, fix a root-leaf path $\Pi = (R_1, R_2, \ldots, R_t)$, and the corresponding collection of shortest-path separators $\P$ of size $\Oh(\log |X|)$. For each $P\ij \in \P$, we may select either the first or the second rounded distance tuple. If, for a path $P\ij$, we select first rounded distance tuple, then, first we need to select points $q_1$ and $q_2$, which leads to $|X|^2$ choices. Then, there are at most $\Oh((z/\mu)^4)$ landmark points along the path, and each distance can take at most $\Oh(z^4/\epsilon^4)$ values. On the other hand, if we select the second rounded distance tuple, we need to select $q_1$, $q_3$ and $q_4$, which leads to $|X|^3$ choices. Thus, for a particular $P\ij$, the number of choices is upper bounded by $|X|^{\Oh(1)} \cdot (z/\epsilon)^{\Oh(z^4/\epsilon^4)}$. Thus, the total number of rounded distance tuples is upper bounded $|X|^{\Oh(\log |X|)} \cdot (z/\epsilon)^{\Oh(z^4 \log |X|/\epsilon^4)}$. This is also the upper bound on the size of $\bbCL$.
	
	Therefore, the size of $\bbC$ is upper bounded by $\exp \lr{ \Oh (\log^2 |X| + z^8 \epsilon^{-8} \log |X| \log(z/\epsilon) ) }$. Finally, we account for rescaling of $\epsilon$ by $\epsilon/(cz\log(z/\epsilon))$, which implies that the actual bound on $|\bbC|$ is:
	$$ \exp \lr{ \Oh \lr { \log^2 |X| + z^{16} \epsilon^{-8} (\log(z/\epsilon))^{8} \log|X|  } }.$$
	
\end{proof}

\section{Applications to Geometric Intersection Graphs}
\label{sec:applications}
In this section, we apply  \Cref{cor:coreset} on various geometric intersection graphs in order to obtain coresets of size that is independent of $n$. First, we consider the case of Euclidean weighted UDG metrics and $\ell_\infty$-weighted USGs, and explain why they satisfy the two canonical properties. Subsquently, we consider UDG and USG with general $\ell_p$ norm weights (\Cref{subsec:other-norms}). Finally, we consider the special case of unweighted UDGs with bounded degree (\Cref{subsec:hop-udg}), where we only discuss the modifications in the arguments required to see that the corresponding metrics also satisfy these properties.

\subsection{Euclidean Weighted Unit Disk Graphs.} \label{subsec:udg}
It is straightforward to verify that Euclidean weighted UDGs satisfy \textsl{Locally Euclidean} property. Indeed, consider an embedding $\lambda: V(G) \to \real^2$ realizing a UDG $G$. There is an edge between two vertices if and only if the euclidean distance between the corresponding two points is at most $2$. Furthermore, the weight of such an edge is exactly the euclidean distance between the two points. It follows that $c_1 = c_2 = 2$ and $c_3 = c_4 = 1$ satisfies the property.

As for the \textsl{Planar Spanner} property, we use the following known result.
\begin{proposition} \label{prop:udg-spanner}
	[Li, Calinescu, and Wan \cite{li2002distributed}] For any UDG $G$, there is a planar spanner $H$ such that for any $u, v \in V(H)$, $\d_{G}(u, v) \le \d_{H}(u, v) \le 2.42 \cdot \d_{G}(u, v)$. 
\end{proposition}

Thus, \Cref{cor:coreset} yields the following theorem for Euclidean-weighted UDGs.
\begin{theorem}\label{thm:UDGscoreset}
Consider the metric space $(V,\d_G)$ induced by any Euclidean weighted unit-disk graph $G=(V,E)$, a set  $P\subseteq V$ with $n$ distinct points, and two positive integers $k$ and $z$. Then there exists a polynomial time algorithm that constructs with probability at least $1-\delta$ a coreset for $(k,z)$-clustering on $P$ of size $\Oh(\epsilon^{-\beta}k\log^2k \log^3 (1/\delta))$, where $z$ is a constant, $\beta = \Oh(z \log z)$, and $\delta < 1/4$. 
\end{theorem}

\subsection{$\ell_\infty$-Weighted Unit Square Graphs} \label{subsec:usg}

\medskip\noindent\textbf{Preliminaries.} Consider a point $q$ in the plane with coordinates $(a, b)$. For a $p \ge 1$, let $||(a, b)||_p \coloneqq \lr{|a|^p+|b|^p}^{1/p}$ for $p \le \infty$, and $||(a, b)||_\infty \coloneqq \lim_{p \to \infty} ||(a, b)||_p = \max\LR{|a|, |b|}$. We have the following relation between different norms.
\begin{observation} \label{obs:lp-norms}
	For any $(a, b) \in \real^2$ and for any $1 \le p \le q \le \infty$, it holds that $||(a, b)||_p \le ||(a, b)||_q \le \sqrt{2} \cdot ||(a, b)||_p$.
\end{observation}
A \emph{unit square} centered at $p$, denoted by $S(p)$, is the axis-parallel square of sidelength $2$ with $p$ at its center, along with its interior.\footnote{Henceforth, whenever we refer to a square, we always refer to the points on the boundary as well as its interior.} Alternatively, $S(p)$ is the set of points $q$ such that $||p - q||_{\infty} \le 1$, where $||\cdot||_{\infty}$ denotes the $\ell_\infty$ norm, defined as $||(a, b)||_{\infty} = \max\{|a|, |b|\}$. 

A Unit Square Graph (USG) with a set of points $P \subset \real^2$ is the intersection graph of unit squares centered at each point in $P$. Observe that $p, q \in P$, $\{p, q\}$ is an edge in the USG if and only if the unit squares $S(p)$ and $S(q)$ intersect.
The weight of an edge $\{p, q\}$ is defined to be $||p-q||_{\infty}$. Observe that the weights of the edges are between $0$ and $2$. 

\paragraph{Canonical Geometric Properties.} First, we argue that $\ell_\infty$-weighted USGs satisfy the \textsl{Locally Euclidean} property. Consider a USG $G$ with an embedding $\lambda: V(G) \to \real^2$. Via \Cref{obs:lp-norms}, it follows that if for some $u, v \in V(G)$, if $||\lambda(u) - \lambda(v)||_2 = |\lambda(u) \lambda(v)| > 2$, then $||\lambda(u) - \lambda(v)||_\infty > 2$, which implies that $uv \not\in E(G)$. On the other hand, if $||\lambda(u) - \lambda(v)||_2 = |\lambda(u) \lambda(v)| \le \sqrt{2}$, then $||\lambda(u) - \lambda(v)||_{\infty} \le 2$, which implies that $uv \in E(G)$. Furthermore, the weight of such an edge is $w(uv) = ||\lambda(u) - \lambda(v)||_{\infty}$, which implies that $|\lambda(u) \lambda(v)| \le w(u, v) \le \sqrt{2} \cdot |\lambda(u) \lambda(v)|$. Thus, $\ell_\infty$-weighted USGs satisfy \textsl{Locally Euclidean} property with $c_1 = \sqrt{2}, c_2 = 2, c_3 = 1, c_4 = \sqrt{2}$. 

Recall that a constant planar spanner for euclidean edge-weighted UDGs was shown in \cite{li2002distributed}. However, a similar spanner for $L_\infty$ edge-weighted USGs was not known before. We bridge this gap by showing the existence of exactly such a spanner in the following theorem.  \Cref{thm:usg-spanner} is interesting on its own and we expect that, similar to planar spanners for UDGs, planar spanners for USGs would find applications beyond clustering. 
To keep the flow of the paper, we defer the proof of  \Cref{thm:usg-spanner} 
to the following section, \Cref{sec:usg-spanner}. 

\begin{restatable}{theorem}{usgspanner} \label{thm:usg-spanner}
	Let $P$ be a set of points satisfying the general position assumptions, and let $G = (P, E)$ be the unit square graph associated with $P$, such that the weight of an edge $uv$ is equal to $\nifty{u- v}$. Then, there exists a planar subgraph $H$ of $G$ such that for any two points $a, b \in P$, $\d_G(a, b) \le \d_H(a, b) \le 3\d_G(a, b)$.
\end{restatable}
Thus, USGs satisfy the planar spanner property with stretch factor $\alpha = 3$. Then, by applying \Cref{thm:coreset}, we obtain the following theorem.
\begin{theorem} \label{thm:usg-coreset}
	Consider the metric space $(V,\d_G)$ induced by any $L_\infty$ weighted unit square graph $G=(V,E)$, a set  $P\subseteq V$ with $n$ distinct points, and two positive integers $k$ and $z$. Then there exists a polynomial-time algorithm that constructs with probability at least $1-\delta$ a coreset for $(k,z)$-clustering on $P$ of size $\Oh(\epsilon^{-\beta}k\log^2k \log^3 (1/\delta))$, where $z$ is a constant, $\beta = \Oh(z \log z)$, and $\delta < 1/4$.  
\end{theorem}

\subsection{UDGs and USGs with Other Norms} \label{subsec:other-norms}

Let $G = (V, E)$ be a UDG corresponding to a set of points $V$ in the plane. Recall that for $u, v \in V$, there is an edge $\{u, v\} \in E(G$) iff $||u-v||_2 \le 2$. In the euclidean edge-weighted UDG, we defined the weight of this edge to also be $||u-v||_2$. However, it is possible to define the weights to be the $\ell_p$ distance between the points for arbitrary $p \ge 1$. To show that $\ell_p$-weighted UDGs satisfy the \textsl{Locally Euclidean} property, we can use arguments similar to the $\ell_\infty$-weighted USG case. We omit the details.

\textbf{Planar Spanner.} Consider the spanner $H'$ obtained by applying \Cref{prop:udg-spanner} to $G'$, where $G'$ is the euclidean (i.e., $L_2$) weighted UDG corresponding to the set of points $V$.  follows that for any $u, v \in V$, $\d_{G'}(u, v) \le \d_{H'}(u, v) \le 2.42 \cdot \d_{G'}(u, v)$. In particular, for any edge $u, v \in E(G)$, $||u-v||_2 \le \d_{H'}(u, v) \le 2.42 \cdot ||u-v||_2$. 

First, consider the case when the weights in $G$ are given by the $L_p$ norm distances between the points, where $1 \le p < 2$. From \Cref{obs:lp-norms}, it follows that $\d_G(u, v) = ||u-v||_p \le ||u-v||_2 \le \d_{H'}(u, v) \le 2.42 \cdot ||u-v||_2 \le 2.42 \cdot \sqrt{2} \cdot ||u-v||_p$. It follows that for any $u', v' \in G$, $\d_{G}(u', v') \le \d_{H'}(u', v') \le 3.42 \cdot \d_G(u', v')$.

Otherwise, the weights in $G$ are given by $L_p$ norm distances, where $2 < p \le \infty$. Then, let $H$ be the weighted graph obtained by multiplying the weight of each edge in the spanner $H'$ by a factor of $\sqrt{2}$. Note that for any edge $\{u, v\} \in E(G)$, \Cref{obs:lp-norms} implies that $||u-v||_p \le \sqrt{2} \cdot ||u-v||_2 \le \sqrt{2} \cdot ||u-v||_p$. This implies that, for $\{u, v\} \in E(G)$,  $\d_G(u, v) = ||u-v||_p \le \sqrt{2} ||u-v||_2 \le \sqrt{2} \cdot \d_{H'}(u, v) = \d_{H}(u, v) \le 2.42 \cdot \sqrt{2} \cdot \sqrt{2} \cdot ||u-v||_p$. This implies that for any $u', v' \in V$, it holds that $\d_{G}(u, v) \le \d_{H}(u', v') \le 4.84 \cdot \d_{G}(u, v)$. 

In either case, we conclude that \Cref{prop:udg-spanner} can be used to show the existence of an $\alpha$-stretch planar spanner, where $\alpha \le 4.84$. 

From the above discussion, it follows that $\ell_p$-weighted UDGs also satisfy the two properties required to apply our framework. Finally, it is easy to modify the previous arguments to also show that $\ell_p$-weighted USGs satisfy the canonical geometric properties. Thus, we conclude the following theorem.
\begin{theorem} \label{thm:udg-usg-coreset-lp}
	Consider the metric space $(V,\d_G)$ induced by any $\ell_p$-weighted UDG (resp.\ $\ell_p$-weighted USG) $G=(V,E)$ for some $1 \le p \le \infty$, a set  $P\subseteq V$ with $n$ distinct points, and two positive integers $k$ and $z$. Then there exists a polynomial-time algorithm that constructs with probability at least $1-\delta$ a coreset for $(k,z)$-clustering on $P$ of size $\Oh(\epsilon^{-\beta}k\log^2k \log^3 (1/\delta))$, where $z$ and $\beta$ are constants, and $\delta < 1/4$.  
\end{theorem}

\subsection{Unweighted Unit Disk Graphs with Bounded Degree} \label{subsec:hop-udg}

Here, we consider metrics induced by \emph{unweighted} UDGs. That is, if $G = (V, E)$ is a Unit Disk Graph, then for any two vertices $u, v \in V(G)$, the distance $\d_{G}(u, v)$ is given by the \emph{hop-length} $|\pi_G(u, v)|$  of a shortest path $\pi_G(u, v)$. Note two points that are very close to each other in an embedding of $G$, are still are at distance $1$ apart according to the hop metric. Therefore, unweighted UDGs do not satisfy \textsc{Locally Euclidean} property. However, they do admit a constant stretch planar spanner, as noted below. In the following, we show that the arguments relying on the locally euclidean properties can be modified for unweighted UDGs, when the maximum degree is upper bounded by $\Delta$. Consequently, the size of the coreset we construct, then, will depend on $\Delta$.

Let $G = (V, E)$ be an unweighted UDG. First, we observe that the distance between any pair of vertices is given by the hop-length of the shortest path between them, and thus is always a non-negative integer. Furthermore, $\d_G(u, v) = 0$ iff $u = v$; otherwise $\d_G(u, v) \ge 1$. 

By examining the proof of the centroid set theorem, we observe the following. \textsl{Locally Euclidean} property is used to construct the sets $\bbCN$, and $\bbCS$, and \textsl{Planar Spanner} property is used to construct the set $\bbCL$. Finally, the centroid set $\bbC$ is defined as the union of the three sets $\bbCN, \bbCS$, and $\bbCL$. In the following, we adapt the construction of $\bbC$ for the unweighted UDGs.

First, we make use of the following result.

\begin{proposition}[\cite{Biniaz20}] \label{prop:hop-spanner}
	For any unweighted UDG $G'$, there exists a planar spanner $H'$ for hop-distances in $G'$. That is, for any $u, v \in V(G')$, $\d_{G'}(u, v) \le \d_{H'}(u, v) \le \alpha \cdot \d_{G'}(u, v)$ for some absolute constant $\alpha \ge 1$.
\end{proposition}

Note that the stretch factor $\alpha$ is an absolute constant, and this result \emph{does not} require that the maximum degree of the UDG is bounded. Following this, we let $\ell \coloneqq \alpha/\epsilon$. We construct the set $\bbCL$ exactly as in the original construction using the recursive decomposition and the planar spanner. 

We define $\bbCN \coloneqq \emptyset$, and construct $\bbCS$ as follows. For every point $p \in X$, we add \emph{all} vertices within $\ell$ hops from $p$ in $G$ to the set $\bbCS$. Alternatively, this can be viewed as an alternative to \Cref{prop:support-graph} to construct a bounded degree support graph, namely the support graph $H$ is equal to $G$, and $f: V(G) \to V(H)$ being the identity mapping. The only caveat is that the degree of $H$ is bounded by $\Delta$, rather than $\Oh(1/\mu^2)$ as required in the original definition. Finally, let $\bbC = \bbCL \cup \bbCS$. We have the following observation.

\begin{claim} \label{cl:support-bounded-degree}
	$|\bbC| = \exp 
	\lr{\Oh \lr{ \log^2 |X| + z^{16} \epsilon^{-8} (\log(z/\epsilon))^{8} \log|X| + \log \Delta}}$
\end{claim}
\begin{proof}
	Since $\bbCL$ is constructed in exactly the same manner as before, we can use the same upper bound as in the proof of \Cref{lem:centroid-size}. \\That is, $|\bbCL| = \exp \lr{\log^2 |X| + z^{16} \epsilon^{-8} (\log(z/\epsilon))^{8} \log|X|}$.
	
	Now we separately bound $|\bbCS|$. Consider a realization of $G$ in the plane. We overlay a grid of sidelength $1$ on the plane. Note that any two points that lie in the same grid cell must have an edge between them. Since the maximum degree of $G$ is $\Delta$, it follows that every cell contains at most $\Delta$ vertices of $V(G)$. Let $H$ denote the 
	Furthermore, if $\{u, v\} \in E(G)$, then the euclidean distance between the centers of the cells containing $u$ and $v$ respectively, is at most $3$. 
	
	Now consider a point $p \in X$, and any vertex $q$ such that $\d_G(p, q) \le \ell$. It follows that the euclidean distance between centers of the cells containing $p$ and $q$ is at most $3\ell$. Therefore, $q$ can belong to one of at most $\pi(3\ell)^2 = \Oh(\ell^2)$ cells satisfying this condition. Since each of the cells contains at most $\Delta$ points of $V(G)$, and $\ell = \Oh(z/\epsilon)$, we add at most $\Oh(z^2\Delta/\epsilon^2)$ points to $\bbCS$ for every point $p \in X$, which implies that $|\bbCS| = \Oh(|X| \cdot z^2\Delta/\epsilon^2)$. 
\end{proof}

Now we discuss the modifications required to construct $\tils$ and in the error analysis. Consider some $s \in \mathcal{S}$. If $s \in \bbCS$, then we let $\tilde{s} = s$, and add it to $\tilS$. Clearly, $\d_G(p, s) = \d_G(p, \tils)$. Otherwise, it follows that for all points $p \in X$, $\d_G(p, s) > \ell$. Then, we use recursive decomposition and case analysis to find a point $\tils \in \bbCL$ that has same rounded distances w.r.t. a landmark set $\mathcal{L}(\cdot)$ as in the original proof. Then, using similar analysis, we can show that for any $p \in X_s$, it holds that $\cost(p, \tilS) \le (1+\epsilon) \cdot \cost(p, \mathcal{S}) + \epsilon \cdot \cost(p, \mathcal{A})$. We omit the details, and conclude with the following theorem.

\begin{theorem} \label{thm:udg-hop-coreset}
	Consider the metric space $(V,\d_G)$ induced by an unweighted UDG $G=(V,E)$ with maximum degree $\Delta$, a set  $P\subseteq V$ with $n$ distinct points, and two positive integers $k$ and $z$. Then there exists a polynomial-time algorithm that constructs with probability at least $1-\delta$ a coreset for $(k,z)$-clustering on $P$ of size $\Oh(\epsilon^{-\beta}k\log^2k \log^3 (1/\delta) + \log(\Delta))$, where $z$ and $\beta$ are constants, and $\delta < 1/4$.  
\end{theorem}


\section{Constant Stretch Planar Spanner for Unit Square Graphs}\label{sec:usg-spanner}

In this section we prove  \Cref{thm:usg-spanner}.
 
\medskip\noindent\textbf{Preliminaries.} 
For points $p = (x_p, y_p), q = (x_q, y_q) \in \real^2$, let $\dx(p,q) = |x_p - x_q|$ and $\dy(p, q) = |y_p - y_q|$, and $D(p, q) = \nifty{p - q} = \max\{\dx(p, q), \dy(p, q)\}$, and $\delta(p, q) = \min\{ \dx(p, q), \dy(p, q) \}$. For points $p$ and $q$, we use $\seg{pq}$ (or $\seg{qp})$ to denote the (closed) line segment between $a$ and $b$. We denote the (euclidean) length of the segment $\seg{pq}$ by $|pq|$.

In the following, whenever we say \emph{square} or a \emph{rectangle}, we always refer to an \emph{axis-parallel square} or an \emph{axis-parallel rectangle}, along with its interior (unless explicitly mentioned otherwise). We denote the sidelength of a square $S$ by $\ell(S)$. A \emph{unit square} centered at a point $p = (x_p, y_p)$, is $S(p) \coloneqq \{q = (x_q, y_q) : D(p, q) \le 1 \}$. Note that $\ell(S(p)) = 2$. Note that for any point $q$ lying on the boundary of $S(p)$, it holds that $D(p, q) = 1$. For convenience, we refer to the four sides of the square corresponding to the cardinal directions as N, E, S, W respectively. For a pair of points $a, b$ such that $\seg{ab}$ is not axis-parallel, we use $R(a, b)$ to denote an axis-parallel rectangle that has $a$ and $b$ as the endpoints of a diagonal.

The following two propositions can be proved via elementary geometry and properties of $\ell_\infty$ distances.

\begin{proposition} \label{prop:sqlen}
	Let $a = (x_a, y_a)$ and $b = (x_b, y_b)$ be arbitrary and distinct points in $\real^2$. Let $S$ be any square having $a$ and $b$ on its boundary. Then, $\ell(S) \ge \difty{a, b}$.
\end{proposition}
\begin{proposition} \label{prop:sqineq}
	Let $a$ and $b$ be two distinct points in the plane such that $\seg{ab}$ is not parallel to $x$ or $y$ axis. Then, for any point $c$ lying on or inside $R(a, b)$, the following holds: $\difty{a,c} \le \difty{a, b}$ and $\difty{b, c} \le \difty{a, b}$.
\end{proposition}

We follow the construction of Binichon et al. \cite{BonichonGHP15}, and show that the $\ell_\infty$ Delaunay Triangulation, when restricted to USG edges, remains a constant stretch spanner for the corresponding USG. 

Fix a set of points $P$. We make the general position assumption, namely, no four points of $P$ lie on the boundary of an axis-parallel square. Furthermore, by slightly perturbing the set of points if necessary, we assume that all pairwise $L_{\infty}$ distances are unique. Under this assumption, the $\ell_\infty$ Delaunay triangulation is unique and planar. If $pq$ is an edge in the Delaunay Triangulation $T$, we say that $pq$ is a \emph{DT} edge. Finally, let $H$ be a subgraph of USG on $P$, defined as follows: for $p, q \in P$, $pq \in E(H)$ iff $pq$ is a DT edge \emph{and} $pq \in E(G)$, i.e., $pq$ is an USG edge. Furthermore, the length of each edge $pq \in E(H)$ is set to be $\difty{pq}$, i.e., if $pq \in E(H)$, then $\d_H(p, q) = \d_G(p, q) = \difty{p, q}$.

In the following, we show the following key lemma.

\begin{lemma} \label{lem:usg-spanner}
	Let $a = (x_a, y_a)$ and $b = (x_b, y_b)$ be two points in $P$ such that $ab \in E(G)$. 
	\\Then, there exists a path $\pi_H(a, b) = (a = v_0, v_1, \ldots, v_t = b)$ in $H$ satisfying the following properties.
	\begin{itemize}
		\item $\d_H(a, b) \le \sum_{i = 1}^{t} \d_H(v_{i-1}, v_{i}) \le 2D(a, b) + \delta(a, b)  \le 3 \cdot \d_G(a, b)$, and
		\item For every $1 \le i \le t$, $\difty{v_{i-1}, v_i} = \d_H(v_{i-1}, v_i) \le \difty{ab}$.
	\end{itemize}
\end{lemma} 

In the following, we fix a pair of points $a, b \in P$ such that $ab$ is not an edge, and follow a geometric construction from \cite{BonichonGHP15} using the Delaunay triangulation of points. We will perform a more careful analysis of this set of points so that subsequently these properties can be used later to prove \Cref{lem:usg-spanner}.

We proceed as in \cite{BonichonGHP15} and let $a = (0, 0)$ and $b = (w, h)$ be a pair of points, such that $ab$ is not a DT edge. Here, we assume that $0 \le h \le w$ and $0 < w$. Thus, note that $\difty{a, b} = w \ge h$. Let $T_1, T_2, \ldots, T_k$ be the sequence of triangles that $\seg{ab}$ intersects when moving from $a$ towards $b$. Suppose the rectangle $R(a, b)$ contains no point of $P$ other than $a$ and $b$. Let $h_0 = l_0 = a$. For each triangle $T_i$, $1 \le i \le k$, the segment $\seg{ab}$ intersects the sides of $T_i$ twice. Let $h_i$ and $l_i$ be the endpoints of the sides of $T_i$ \emph{last}, while moving from $a$ towards $b$, such that $h_i$ lies above $\seg{ab}$, and $l_i$ lies below $\seg{ab}$ (here, above and below are defined using the two half-planes defined by the line passing through $a$ and $b$). Note that either $h_{i-1} = h_i$ and $T_i = \Delta(h_i, l_i, l_{i-1})$, or $l_{i-1} = l_i$ and $T_i = \Delta(h_{i-1}, h_i, l_i)$ for $1 < i < k$. Recall that $T_1 = \Delta(a, h_1, l_1)$, and $T_k = \Delta(h_k, l_k, b)$. For $1 \le i \le k$, let $S_i$ to be the empty square having vertices of $T_i$ on its boundary (recall that the empty square assumption follows from the fact that the corresponding edges are DT edges). Finally, say that a point $u = (x_u, y_u)$ is \emph{high} (resp.\ \emph{low}) w.r.t. $R(a, b)$ if $0 \le x_u \le w$ and $y_u > h$ (resp.\ $y_u < 0$). Bonichon et al.~\cite{BonichonGHP15} prove the following lemma.

\begin{figure}
	\centering
	\includegraphics{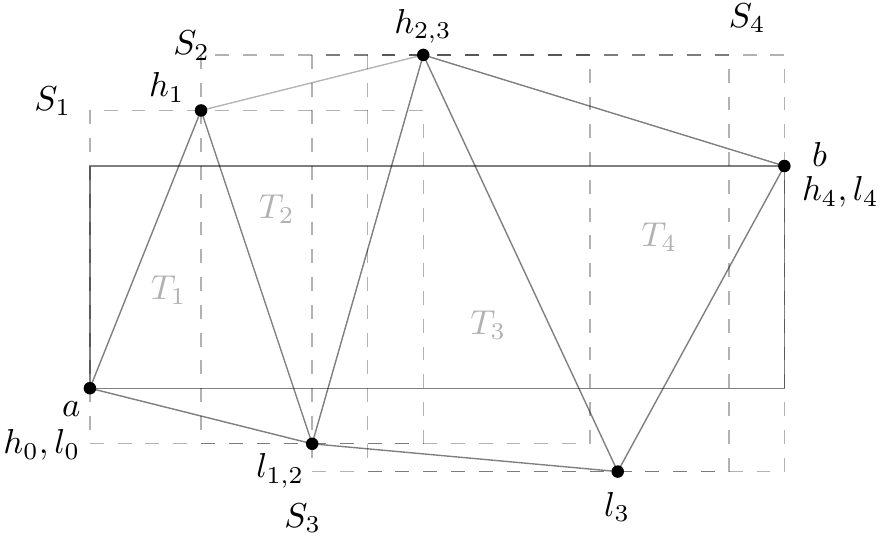}
	\caption{Various points, triangles, and squares used in the construction.} \label{fig:triangles}
\end{figure}

\begin{lemma}[Lemma 7 in \cite{BonichonGHP15}] \label{lem:bonichon-lemma}
	If $ab$ is not a DT edge, and if rectangle $R(a, b)$ contains no point of $P$ other than $a$ and $b$, then the following properties hold.
	\begin{enumerate}
		\item $a$ lies on the W side of $S_1$,
		\item $b$ lies on the E side of $S_k$,
		\item Points $h_1, \ldots, h_{k-1}$ are high, and points $l_1, \ldots, l_{k-1}$ are low w.r.t. $R(a, b)$, and
		\item For any $1 < i < k$, 
		\begin{itemize}
			\item Either $T_i = \Delta(h_{i-1}, h_i, l_{i-1} = l_i)$, and points $h_{i-1}, h_i$ and $l_{i-1} = l_i$ lie on the sides of $S_i$ in clockwise order with no two points on the same side, and $(h_{i-1}, h_i)$ is a WN, WE or NE edge in $S_i$, or
			\item  $T_i = \Delta(h_{i-1} = h_i, l_{i-1}, l_i)$, and points $h_{i-1} = h_i, l_{i}$, and $l_{i-1}$ lie on the sides of $S_i$ in clockwise order with no two points on the same side, and $(l_{i-1}, l_i)$ is a WS, WE or SE edge in $S_i$.
		\end{itemize}
	\end{enumerate}
\end{lemma}

Using the previous lemma, we prove the following two properties.

\begin{claim} \label{cl:hi-li-distances}
	For any $1 \le i \le k$, $\max \LR{\difty{h_i, h_{i-1}}, \difty{l_i, l_{i-1}}, \difty{h_i, l_i}} < w = \difty{a, b}$.
\end{claim}
\begin{proof}
	First, consider triangle $T_1 = \Delta(a, h_1, l_1)$. Suppose there is some side $pq$ of $T_1$ such that $\difty{pq} \ge w$. Then, \Cref{prop:sqlen} implies that $\ell(S_1) \ge \difty{pq} \ge w$. Now, from properties 1 and 3 of \Cref{lem:bonichon-lemma}, $a$ lies on the W side of $S_1$, and $h_1$ (resp.\ $l_1$) is high (resp.\ low) w.r.t. $R(a, b)$. Therefore, point $b$ is contained in $S_1$. However, this is a contradiction, since $S_1$ is defined by three DT edges, and is thus empty. The case for the triangle $T_k = \Delta(h_k, l_k, b)$ is analogous.
	
	Now consider a triangle $T_i$ with $1 < i < k$, and suppose $T_i = \Delta(h_{i-1}, h_i, l_{i-1} = l_i)$ (the other case is analogous). Note that none of the vertices of $T_i$ is either $a$ or $b$. 	Again, suppose for some side $pq$ of $T_i$, it holds hat $\difty{pq} \ge w$. Again, by \Cref{prop:sqlen}, it follows that $\ell(S_i) \ge \difty{pq} \ge w$. Recall that from property 3 of \Cref{lem:bonichon-lemma}, $h_i$ (resp.\ $l_i$) is high (resp.\ low) w.r.t. $R(a, b)$, i.e., the $x$-coordinates of $h_i$ and $l_i$ lie within the interval $[0, w]$. Since $h_i$ and $l_i$ lie on the boundary of a square of sidelength at least $w$, this implies that $S_i$ must contain either $a$ or $b$, which contradicts the empty square property for $S_i$. 
\end{proof}

\begin{claim} \label{cl:hili-bounded}
	For any $1 \le i \le k$, $D(a, h_i) < D(a, b)$, and $D(a, l_i) < D(a, b)$.
\end{claim}
\begin{proof}
	Fix some $1 \le i \le k$, and first suppose for contradiction that $D(a, h_i) > D(a, b) = w$. Now, $D(a, h_i) = \max\{|x_{h_i}|, |y_{h_i}|\}$. However, from \Cref{lem:bonichon-lemma}, we know that $h_i$ is high, i.e., $0 \le x_{h_i} \le w$, and $y_{h_i} > h$, as well as $l_i$ is low, i.e., $y_{l_i} < 0$. Therefore, $D(a, h_i) = y_{h_i} > w$. Now, consider $D(h_i, l_i) \ge y_{h_i} - y_{l_i} > w$, since $y_{l_i} < 0$. However, this is a contradiction to \Cref{cl:hi-li-distances}, which shows that $D(h_i, l_i) < w$.
	
	Now, suppose $D(a, l_i) > D(a, b) = w$. Now, $D(a, l_i) = \max\{ |x_{l_i}|, |y_{l_i}| \} = |y_{l_i}|$. This implies that $y_{l_i} < -w$. Now, consider $D(h_i, l_i) \ge y_{h_i} - y_{l_i} > h + w > w$, which contradicts \Cref{cl:hi-li-distances}. 
\end{proof}

Now we introduce further definitions and terminology. A vertex $c$ of a triangle $T_i$ is \emph{eastern in $S_i$} if it lies on the E side of $S_i$. An edge in $T$ is said to be \emph{gentle} if the corresponding line segment has slope within $[-1, 1]$; otherwise we say that it is \emph{steep}. By \Cref{lem:bonichon-lemma} and from the general position assumption, it follows that if an edge $l_j h_j$ in $T$ is gentle, then $l_j$ or $h_j$ must be eastern in $S_j$. 

\begin{definition}
	Suppose $(a, b)$ is not a DT edge, and rectangle $R(a, b)$ contains no point of $P$ other than $a$ and $b$. We say that a square $S_j$, $1 \le j \le k$, is \emph{inductive} if the edge $l_j h_j$ is gentle. The eastern point $c = h_j$ or $c = l_j$ is the \emph{inductive point} of $S_j$.
\end{definition}

Now, we have the following lemmas from Bonichon et al.\ \cite{BonichonGHP15}.
\begin{lemma}[Lemma 9 from \cite{BonichonGHP15}] \label{lem:hl-path}
	Suppose $ab$ is not a DT edge, and rectangle $R(a, b)$ contains no point of $P$ other than $a$ and $b$. Suppose the coordinates of point $c = h_i$ or $c = l_i$ satisfy $w - x_c < |h - y_c|$. 
	\begin{enumerate}
		\item If $c = h_i$, and thus $w - x_{h_i} < y_{h_i} - h$, then there exists an index $j$ with $i < j \le k$ such that all edges in path $h_i, h_{i+1}, \ldots, h_j$ are NE edges in the respective squares, and $w - x_{h_j} \ge y_{h_j} - h \ge 0$.
		\item If $c = l_i$, and thus $w - x_{l_i} < h - y_{l_i}$, then there exists an index $j$ with $i < j \le k$ such that all edges in path $l_i, l_{i+1}, \ldots, l_j$ are NE edges in the respective squares, and $w - x_{l_j} \ge y_{l_j} - h \ge 0$.
	\end{enumerate}
\end{lemma}

\begin{lemma}[Partial Statement of Lemma 8 from \cite{BonichonGHP15}] \label{lem:ind-path}
	Suppose $(a, b)$ is not a DT edge, and rectangle $R(a, b)$ contains no point of $P$ other than $a$ and $b$. If $S_j$ is the first inductive square (if any) in the sequence $S_1, S_2, \ldots, S_{k-1}$. If $h_j$ is the inductive point of $S_j$, then $\d_H(a, h_j) + (y_j - h) \le 2hx_{h_j}$. Otherwise, if $l_j$ is the inductive point of $S_j$, then $\d_H(a, l_j) - y_{l_j} \le 2x_{l_j}$.
\end{lemma}

Now, we are ready to prove \Cref{lem:usg-spanner}.

\begin{proof}[Proof of \Cref{lem:usg-spanner}]
	We prove the lemma by induction on $D(a, b)$. Note that from our assumption, the distances are unique, and we only need to consider pairs $a, b \in P$ such that $D(a, b) \le 2$.
	
	\textbf{Base case.} Consider points $a, b \in P$ that have the smallest $D(a, b)$ value. Suppose $a = (0, 0)$ and $b = (w, h)$ with $0 \le h \le w$ and $0 < w = D(a, b)$. Since $a$ and $b$ are closest points (w.r.t. $\ell_\infty$ distance), the largest square having $a$ as its SW vertex, and containing no other point of $P$ in its interior, must have $b$ on its boundary (via \Cref{prop:sqineq}). Therefore, $ab$ is a DT edge. Furthermore, since $D(a, b) \le 2$, $ab \in E(G)$, which implies that $ab \in E(H)$. Finally, we observe that $\d_H(a, b) = D(a, b) \le w + h \le 3w + h$. 
	
	\textbf{Inductive step.} Consider points $a$ and $b$, and suppose the statement holds for all pairs of points with $\ell_\infty$ distance less than $D(a, b) \le 2$. If $ab \in E(H)$, then we proceed as in the previous case and obtain the result. Therefore, suppose $ab$ is not an edge in $T$.
	
	Again, we assume that $a= (0, 0)$ and $b = (w, h)$ where $0 \le h \le w$ and $0 < w = D(a, b)$ (note that this assumption is wlog, since we can appropriately translate and rotate the plane by an integral multiple of $90^\circ$, and observe neither of these operations affect the $\ell_\infty$ distances).
	
	\textbf{Case 1a.} Suppose there is at least one point of $P$ lying within rectangle $R(a, b)$. If there is a point $c \neq a, b$ such that $0 \le y_c, \le x_c$, $0 < x_c$, $0 \le h- y_c \le w - x_c$, and $0 < w - x_c$. Then, $D(a, c) = x_c < w = D(a, b)$, and $D(b, c) = w - x_c < w = D(a, b)$ region B in \Cref{fig:regions}. 
	Therefore, we use induction hypothesis to obtain that $\d_H(a, c) \le 3x_c + y_c$, and $\d_H(c, b) \le 3 (w-x_c) + (h-y_c)$. Furthermore, the length of each edge in the path $\pi_H(a, c)$ (resp.\ $\pi_H(c, b)$) is at most $D(a, c) < w$ (resp.\ $D(c, b) < w$). By concatenating the two paths, we obtain a path $\pi_H(a, b)$ of length at most $2(x_c + w-x_c) + (y_c + h-y_c) = 2w + h$, such that edge has length at most $w$.
	
	\begin{figure}
		\centering
		\includegraphics[scale=0.5]{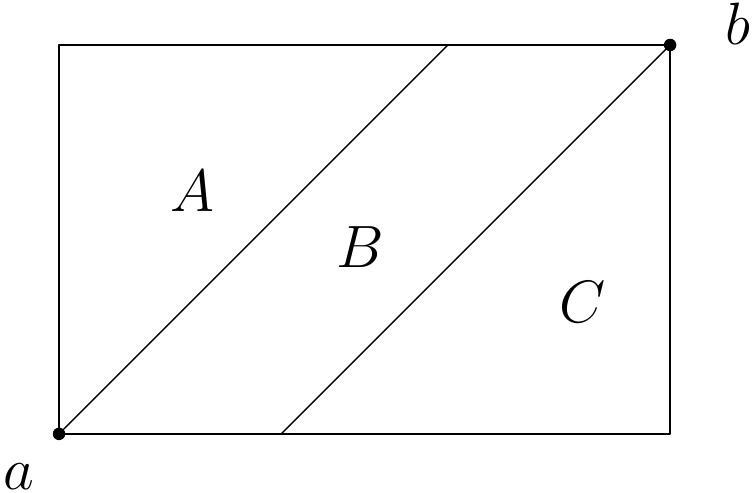}
		\caption{Different regions used in Case 1 in the proof of \Cref{lem:usg-spanner}.} \label{fig:regions}
	\end{figure}
	
	\textbf{Case 1b.} Now, suppose that there is no point in region B, but $R(a, b)$ contains a point of $P$. In this case, let $S_a$ (resp.\ $S_b$) be a square with $a$ as its SW corner (resp.\ $b$ its NE corner) and containing no other point of $P$ in its interior. Then, by assumption, there exists a point $c \neq a, b$ such that $c$ is on the boundary of square $S_a$ or of $S_b$. W.l.o.g. assume former (the other case is analogous). Then, note that $c$ belongs to region A in \Cref{fig:regions}, i.e., $y_c > x_c > 0$, and also $ac$ is a DT edge. Note that $x_c > 0$ since $a$ and $c$ cannot lie on the same side of an empty square by general position assumption. Then, $\d_H(a, c) = D(a, c) \le x_c + y_c$. From \Cref{prop:sqineq}, we also obtain that $D(a, c) \le D(a, b)$.  
	Furthermore, $D(b, c) = w - x_c < w = D(b, c)$. Therefore, we can apply induction hypothesis on the pair $b, c$, to conclude that there exists a path $\pi_H(c, b)$ such that $\d_H(b, c) \le 2 (w-x_c) + (h-y_c)$, and each edge along the path has length at most $D(b, c) < D(a, b)$. Then, by concatenating $ac$ with $\pi_H(c, b)$, we obtain a path of length at most $x_c + y_c + 2 (w-x_c) + h-y_c \le 2w + h$, such that the length of each edge is at most $D(a, b)$.
	
	\textbf{Case 2.} Suppose $R(a, b)$ contains no point of $P$ other than $a$ and $b$.
	
	\textbf{Case 2a.} If no square $S_1, S_2, \ldots, S_{k-1}$ is inductive, then  we use \Cref{lem:potential-path} to obtain that $\d_H(a, b) \le 2w$, and the corresponding path $\pi_H(a, b)$ contains edges $uv$ such that $D(u, v) < w$. This Lemma is proved using induction hypothesis, and we defer its proof after completing the current proof.
	
	Otherwise, let $S_i$ be the first inductive square. Now we consider two cases based on the inductive point of $S_i$. 
	\\\textbf{Case 2b.} Suppose $h_i$ is the inductive point of $S_i$. By \Cref{lem:bonichon-lemma}, $x_{ih} > 0$, and by \Cref{lem:hl-path}, there exists an index $j$ with $i \le j \le k$ such that $\pi_{H}(h_i, h_j) = (h_i, h_{i+1}, \ldots, h_j)$ is a path of length at most $(x_{h_j} - x_{h_i}) + (y_{h_i} - y_{h_j})$. Furthermore, from \Cref{cl:hi-li-distances}, it follows that $D(h_{i'}, h_{i'+1}) < D(a, b)$ for $i \le i' \le j-1$. \Cref{lem:hl-path} also implies that $w - x_{h_j} \ge y_{h_j} - h \ge 0$, and by \Cref{lem:bonichon-lemma}, $x_{h_j} \ge x_{h_i}$. Now, either (1) $h_j = b$, in which case it follows that $\d_H(h_j, b) = 0 \le 2(w-x_{h_j}) + (y_{h_j} - h)$. In this case, we let $\pi_{H}(h_j, b)$ to be a path of length $0$. (2) Otherwise, $h_j \neq b$, which implies that $x_{h_j} \le x_{h_{k-1}} < x_{b} = w$, where the strict inequality is due to the general position assumption. Therefore, $D(h_j, b) = \max\LR{w - x_{h_j}, y_{h_j} - h} = w - x_{h_j} < w = D(a, b)$, and $\delta(h_j, b) = y_{h_j} - h$. Therefore, by induction hypothesis, we obtain that there exists a path $\pi_{H}(h_j, b)$ of length at most $2(w-x_{h_j}) + (y_{h_j} - h)$, such that each edge of the path has length at most $D(a, b) \le 2$. Now, let $\pi_H(a, b)$ be a path obtained by concatenating the paths $\pi_H(a, h_i), \pi_H(h_i, h_j)$, and $\pi_{h_j, b}$ as defined above. It follows that
	\begin{align*}
		\d_H(a, b) &\le \d_H(a, h_i) + \d_H(h_i, h_j) + \d_H(h_j, b)
		\\&\le 2x_{h_i} - (y_{h_i} - h) + (x_{h_j} - x_{h_i}) + (y_{h_i} - y_{h_j}) + 2(w-x_{h_j}) + (y_{h_j} - h)
		\\&\le 2w \tag{Since $x_{h_i} \le x_{h_j}$}
	\end{align*}
	\textbf{Case 2c.} Suppose $l_i$ is the inductive point of $S_i$. By \Cref{lem:bonichon-lemma}, $x_{l_i}> 0$, and by \Cref{lem:hl-path}, there is a $j$ with $i \le j \le k$ such that $\pi_{H}(l_i, l_j) = (l_i, l_{i+1}, \ldots, l_{j})$ is a path of length at most $(x_{l_j} - x_{l_i}) + (y_{l_j} - y_{l_i})$. Furthermore, from \Cref{cl:hi-li-distances}, it follows that $D(l_{i'}, l_{i'+1}) < D(a, b)$ for $i \le i' \le j-1$. \Cref{lem:bonichon-lemma} also implies that $x_{l_j} \ge x_{l_i}$. Then, we consider two cases as before, namely (1) $l_j = b$, or (2) $w - x_{l_j} > 0$. In either case, $D(l_j, b) < D(a, b)$. Then, proceeding analogous to the previous case, we can obtain a path $\pi_H(l_j, b)$ of length at most $2(w-x_{l_j}) + (h-y_{l_j})$ such that each edge in the path has length at most $D(a, b)$. Again, by concatenating the paths $\pi_H(a, l_i), \pi_{H}(l_i, l_j)$, and $\pi_H(l_j, b)$, we obtain a path $\pi_H(a, b)$ such that each edge in the path has length at most $D(a, b) \le 2$. It follows that,
	\begin{align*}
		\d_H(a, b) &\le \d_H(a, l_i) + \d_H(l_i, l_j) + \d_H(l_j, b)
		\\&\le 2x_{l_j} + y_{l_i} + (x_{l_j} - x_{l_i}) + (y_{l_j} - y_{l_i}) + 2(w-x_{l_j}) + (h-y_{l_j})
		\\&\le 2w + h
	\end{align*}
\end{proof}

\begin{lemma} \label{lem:potential-path}
	Let $a, b \in P$ be a pair of points such that (i) $D(a, b) \le 2$, (ii) the rectangle $R(a, b)$ contains no point of $P$ other than $a$ and $b$, and (iii) no square $S_1, S_2, \ldots, S_{k-1}$ is inductive. Then, there exists a path $\pi_{H}(a, b)$ of length at most $2w$, such that each edge $uv$ along the path $\pi_H(a, b)$ satisfies that $D(u, v) < w$.
\end{lemma}
\begin{proof}
	In order to prove this lemma, we need to introduce additional terminology. 
	\\\textbf{Terminology.} Let $\delta_0 = 0$, and for $1 \le i \le k$, let $\delta_i$ be the horizontal distance between point $a$ and the E side of square $S_i$. A square $S_i$ has \emph{potential} if $\d_H(a, h_i) + \d_H(a, l_i) + P_{S_i}(h_i, l_i) \le 4\delta_i$, where $P_{S_i}(h_i, l_i)$ is the length of the path when moving from $h_i$ to $l_i$ along the sides of $S_i$ in clockwise manner. 
	
	From the assumptions of the statement of the lemma, we will show that the squares $S_1, S_2, \ldots, S_{k-1}, S_k$ all have potential, and use this fact to construct a path $\pi_H(a, b)$ with the desired properties.
	
	By \Cref{lem:bonichon-lemma}, $a$ lies on the W side of $S_1$ and $\delta_1$ is the sidelength of square $S_1$. Then, $\d_{H}(a, h_1) + \d_H(a, l_1) + P_{S_1}(h_1, l_1)$ is upper bounded by the perimeter of $S_1$, which is $4\delta_1$.
	
	Now, inductively assume that we have shown that square $S_i$ has potential, i.e., $\d_H(a, h_i) + \d_H(a, l_i) + P_{S_i}(h_i, l_i) \le 4 \delta_i$. Note that from the assumption, $S_i$ is not inductive. Also inductively suppose that the paths $\pi_H(a, h_i)$, $\pi_H(a, l_i)$ witnessing $\d_H(a, h_i)$ and $\d_H(a, l_i)$ respectively. Then, we show that the square $S_{i+1}$ has potential, and construct the corresponding paths $\pi_H(a, h_{i+1})$ and $\pi_H(a, l_{i+1})$ witnessing $\d_H(a, h_{i+1})$ and $\d_H(a, l_{i+1})$ respectively. 
	
	Squares $S_i$ and $S_{i+1}$ both contain points $l_i$ and $h_i$. Since $S_i$ is not inductive, the edge $l_i h_i$ must be steep, i.e., $\dx(l_i, h_i) < \dy(l_i, h_i)$. First, we consider the case when $x_{l_i} < x_{h_i}$, and the case when $x_{l_i} > x_{h_i}$ can be shown analogously.
	
	By \Cref{lem:bonichon-lemma}, $T_i = \Delta(h_{i-1}, h_{i}, l_{i-1} = l_i)$, or $T_{i} = \Delta(h_{i-1} = h_i, l_{i-1}, l_i)$, and there is a side of $S_i$ between the sides on which $l_i$ and $h_i$ lie, when moving clockwise from $l_i$ to $h_i$. From \Cref{lem:bonichon-lemma} and $x_{l_i} < x_{h_i}$, we conclude that $l_i$ lies on the S side, and $h_i$ lies on the N or E side of the square $S_i$.
	
	\begin{figure}
		\centering
		\includegraphics[scale=0.8]{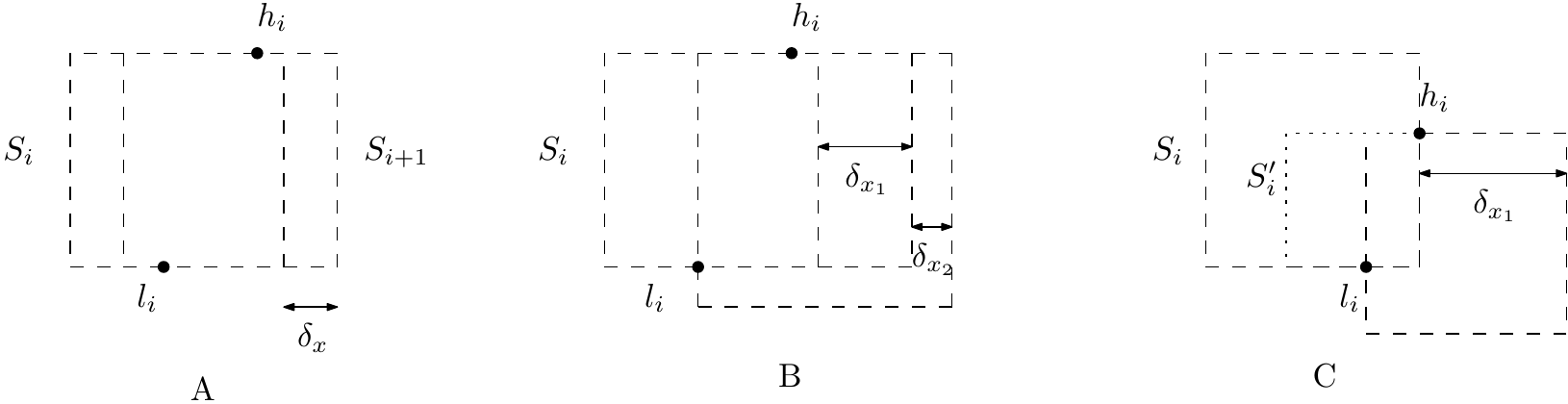}
		\caption{Illustration for the three cases in the proof of \Cref{lem:potential-path}} \label{fig:squares}
	\end{figure}
	
	If $h_i$ is on the N side of $S_i$, and since $x_{l_i} < x_{h_i}$, $h_i$ must also be on the N side of $S_{i+1}$. There are two possibilities for the position of $l_i$ on the boundary of $S_{i+1}$. If $l_i$ is on the S side of $S_{i+1}$, then $S_{i+1}$ is obtained by a horizontal translation of length $\delta_x = \delta_{i+1} - \delta_i$ as shown in \Cref{fig:squares} (A). Then,
	\begin{equation}
		P_{S_{i+1}}(h_i, l_i) - P_{S_i}(h_i, l_i) = 2\delta_x = 2(\delta_{i+1} - \delta_i) \label{eqn:ps-ineq1}
	\end{equation}
	Otherwise, $l_i$ is on the W side of $S_{i+1}$. Let $S'_i$ be the square with $l_i$ as its SW corner, and obtained by a horizontal translation of $S_i$. Let $\delta_{x_1}$ denote the length of the horizontal translation, and $\delta_{x_2}$ denote the difference between the sidelength of $S'_{i}$ and $S_{i+1}$ as shown in \Cref{fig:squares} (B). Then, $\delta_{x_1} + \delta_{x_2} = \delta_{i+1} - \delta_i$, and
	\begin{equation}
		P_{S_{i+1}}(h_i, l_i) - P_{S_i}(h_i, l_i) = 2\delta_{x_1} + 4\delta_{x_2} \le 4(\delta_{i+1} - \delta_i) \label{eqn:ps-ineq2}
	\end{equation}
	If $h_i$ is on the E side of $S_i$, then let $S'_i$ be the square that shares its SE corner with $S_i$, and with $h_i$ as its NE corner, as shown in \Cref{fig:squares} (C). Since $l_i h_i$ is steep, $l_i$ lies on the S side of $S'_i$. Then, using the analysis from the previous case, we obtain that $P_{S_{i+1}}(h_i, l_i) - P_{S'_i}(h_i, l_i) \le 4(\delta_{i+1} - \delta_i)$, and since $P_{S'_{i}}(h_i, l_i) = P_{S_i}(h_i, l_i)$, we conclude that the following inequality holds in all cases.
	\begin{align}
		P_{S_{i+1}}(h_i, l_i) - P_{S_i}(h_i, l_i) \le 4(\delta_{i+1} - \delta_i) \label{eqn:ps-ineq}
	\end{align}
	Now, since $S_i$ has potential, 
	\begin{align}
		&\d_H(a, h_i) + \d_H(a, l_i) + P_{S_{i+1}}(h_i, l_i) \nonumber
		\\&= \d_H(a, h_i) + \d_H(a, l_i) + P_{S_{i}}(h_i, l_i) + (P_{S_{i+1}}(h_i, l_i) -  P_{S_{i}}(h_i, l_i)) \nonumber
		\\&\le 4\delta_i + 4(\delta_{i+1} - \delta_i) = 4\delta_{i+1} \label{eqn:siplus1}
	\end{align}
	Now, suppose that $T_{i+1} = \Delta(h_i, h_{i+1}, l_i = l_{i+1})$. Then, $h_i h_{i+1}$ is an edge of DT. Therefore, $h_i h_{i+1}$ is an edge in $H$. By \Cref{lem:bonichon-lemma}, $h_{i+1}$ lies somewhere on the boundary of $S_{i+1}$ between $h_i$ and $l_i$, when moving clockwise from $h_i$ to $l_i$. Then, by triangle inequality, $D(h_i, h_{i+1}) \le P_{S_{i+1}}(h_i, h_{i+1})$. Then,
	\begin{align*}
		&\d_H(a, h_{i+1}) + \d_{H}(a, l_{i+1}) + P_{S_{i+1}}(h_{i+1}, l_{i+1}) 
		\\&\le \d_H(a, h_i) + D(h_i, h_{i+1}) + \d_H(a, l_i) + P_{S_{i+1}}(h_{i+1}, l_i)
		\\&\le \d_H(a, h_i) + \d_H(a, l_i) + P_{S_{i+1}}(h_i, l_i)
		\\&\le 4\delta_{i+1} \tag{From \Cref{eqn:siplus1}}
	\end{align*}
	Furthermore, by \Cref{cl:hi-li-distances}, $D(h_i, h_{i+1}) < D(a, b)$. Thus, we obtain the path $\pi_{H}(a, h_{i+1})$ by appending the edge $h_i h_{i+1}$ at the end of the inductively constructed path $\pi_{H}(a, h_i)$. Since $l_{i} = l_{i+1}$, the path $\pi_{H}(a, l_{i+1})$ is same as $\pi_{H}(a, l_i)$. The case when $T_{i+1} = \Delta(h_i = h_{i+1}, l_i, l_{i+1})$ is symmetric. Thus, we have completed the inductive step.
	
	Thus, at the end we have shown that the square $S_k$ has a potential, i.e., $\d_{H}(a, h_k) + \d_H(a, l_k) + P_{S_k}(h_k , l_k) \le 4\delta_k = 4w$. Recall that $h_k = l_k = b$, which implies that $P_{S_k}(h_k, l_k) = 0$. Therefore, at least one of the inductively constructed paths $\pi_H(a, b)$ has length at most $2\delta_k = 2w$, and each edge $uv$ on the path satisfies $D(u, v) < w$. This concludes the proof of the lemma.
\end{proof}

We conclude with the proof of the main theorem of the section, which we restate here.

\usgspanner*
\begin{proof}
	Consider two points $a, b \in P$, and let $\pi_G(a, b) = (a = v_0, v_1, v_2, \ldots, v_t = b)$ denote a shortest path in $G$, such that for $0 \le i < t$, $v_i v_{i+1} \in E(G)$ with $\d_G(v_i, v_{i+1}) = D(v_i, v_{i+1})$. Since $H$ is a subgraph of Delaunay Triangulation, it is planar, and by construction it is also a subgraph of $G$. Now, using \Cref{lem:usg-spanner} for each edge $v_i v_{i+1}$, we know that there exists a path $\pi_{H}(v_i, v_{i+1})$ in $H$ of length at most $3D(v_i, v_{i+1}) = 3\d_G(v_{i}, v_{i+1})$. It is easy to see that the path $\pi_{H}(a, b)$ is obtained by concatenating all such paths and short-cutting the edges if necessary, satisfies the conditions in the statement of the theorem.
\end{proof}

\section{Conclusion and Open Questions} \label{sec:conclusion}
We obtain the first coresets for $k$-clustering problems whose size is independent of $n$, on a variety of geometric graph metrics, such as weighted intersection graphs of unit disks and squares. A UDG (or a USG) can contain arbitrarily large cliques, i.e., they can be (locally) dense. Therefore, to the best of our knowledge, ours is the first small-sized (i.e., independent of $n$) coreset construction for a shortest-path metric on a dense family of graphs. Due to the inherently ``hybrid'' nature of such metrics, our coreset construction has to carefully navigate the locally-Euclidean and globally-sparse nature of the metric.

We believe the contribution of our work is also conceptual, in that we ``abstract out'' the geometric structural properties of the metrics that are sufficient to obtain small-sized coresets via the versatile framework of Cohen-Addad et al. \cite{Cohen-AddadSS21}. These structural properties are also satisfied by $\ell_p$-norm weighted UDGs and USGs. Furthermore, by suitably modifying the construction, we can also handle hop metrics (i.e., unweighted edges) induced by UDGs of bounded degree. Thus, we obtain small-sized coresets for $k$-clustering problems for all of these graph families. In order to obtain the result on USGs, we prove that these graphs admit a $3$-stretch planar spanner, a result that may be of independent interest and of further applicability.

The most natural question is to find more examples of geometric intersection graph families that satisfy the structural properties identified in this work. Disk graphs in $\real^2$ and Unit Ball Graphs in $\real^d$ (for constant $d \ge 3$) are two orthogonal generalizations of UDGs, and thus may be the most obvious candidates. However, these graph families are not known to admit a constant stretch planar spanner. As an intermediate step, it might be interesting to consider unit disk graphs on a surface $\Sigma$ of bounded genus. Here, it might be more natural to require whether such a graph admit constant stretch spanner that is also embeddable on $\Sigma$ (which is a relaxation of planarity). It might be possible to extend our framework with this relaxed setting, also yielding smaller coresets for such geometric intersection graph families.

\bibliography{references}
\end{document}